\documentclass[prb,aps,twocolumn,longbibliography,superscriptaddress]{revtex4-1}
\usepackage{graphicx}
\usepackage{amsmath}
\usepackage{xcolor}
\usepackage{caption}
\usepackage{subcaption}
\usepackage{bm}
\usepackage{amsthm}
\usepackage{esint}
\usepackage{amsfonts}
\usepackage{appendix}

\newtheorem{theorem}{Theorem}
\newtheorem{lemma}{Lemma}
\newtheorem{corollary}{Corollary}
\newcommand{\mr}{moir\'e~}

\newcommand{\ourinv}{intra-valley inversion~}

\newcommand{\Blochk}{Bloch momentum~}
\newcommand{\Blochkk}{Bloch momentum}
\newcommand{\p}{\partial}
\newcommand{\beqn}{\begin{eqnarray}}
\newcommand{\eeqn}{\end{eqnarray}}
\newcommand{\mG}{\mathcal{G}}
\newcommand{\mH}{\mathcal{H}}
\newcommand{\mD}{\mathcal{D}}
\newcommand{\mC}{\mathcal{C}}
\newcommand{\mT}{\mathcal{T}}
\newcommand{\mM}{\mathcal{M}}
\newcommand{\mI}{\mathcal{I}}
\newcommand{\diag}{\text{diag}}
\newcommand{\ts}{\textsuperscript}

\begin{document}

\title{Chiral Approximation to Twisted Bilayer Graphene: Exact Intra-Valley Inversion Symmetry, Nodal Structure and Implications for Higher Magic Angles}

\author{Jie Wang}
\email{jiewang@flatironinstitute.org}
\affiliation{Center for Computational Quantum Physics, Flatiron Institute, 162 5th Avenue, New York, NY 10010, USA}
\author{Yunqin Zheng}
\affiliation{Institute for Solid State Physics, University of Tokyo, Kashiwa, Chiba 277-8581, Japan}
\affiliation{Kavli Institute for the Physics and Mathematics of the Universe,
University of Tokyo, Kashiwa, Chiba 277-8583, Japan}
\author{Andrew J. Millis}
\email{amillis@flatironinstitute.org}
\affiliation{Center for Computational Quantum Physics, Flatiron Institute, 162 5th Avenue, New York, NY 10010, USA}
\affiliation{Department of Physics, Columbia University, 538 W 120th Street, New York, New York 10027, USA}
\author{Jennifer Cano}
\email{jcano-affiliate@flatironinstitute.org}
\affiliation{Center for Computational Quantum Physics, Flatiron Institute, 162 5th Avenue, New York, NY 10010, USA}
\affiliation{Department of Physics and Astronomy, Stony Brook University, Stony Brook, New York 11974, USA}

\begin{abstract}
This paper presents a mathematical and numerical analysis of the flatband wavefunctions occurring in the chiral model of twisted bilayer graphene at the ``magic'' twist angles. We show that the chiral model possesses an exact intra-valley inversion symmetry. Writing the flatband wavefunction as a product of a lowest Landau level quantum Hall state and a spinor, we show that the components of the spinor are anti-quantum Hall wavefunctions related by the inversion symmetry operation introduced here. We then show numerically that as one moves from the lowest to higher magic angles, the spinor components of the wavefunction exhibit an increasing number of zeros, resembling the changes in the quantum Hall wavefunction as the Landau level index is increased. The wavefunction zeros are characterized by a chirality, with zeros of the same chirality clustering near the center of the moir\'e unit cell, while opposite chirality zeros are pushed to the boundaries of the unit cell. The enhanced phase winding at higher magic angles suggests an increased circulating current. Physical implications for scanning tunneling spectroscopy, orbital magnetization and interaction effects are discussed.
\end{abstract}

\maketitle
\section{Introduction}
When one graphene layer is stacked on top of another layer with small relative twist angle, a \mr super-lattice pattern is created. At particular twist angles, referred to by Bistritzer and MacDonald as ``magic angles" \cite{Bistritzer12233}, the bands near the chemical potential are dramatically flattened and separated from other bands \cite{Santos,Santos2,Mele_TBG,Bistritzer12233}. Experiments on ``magic angle'' bilayers report interesting phenomena including superconductivity, interaction-driven insulating states and anomalous Hall effects \cite{Cao:2018aa,Cao:2018ab,Serlin900,Sharpe605,Cory_Science19,Young_Naturephys19,Young_Naturephys20,Tang:2020aa,Balents:2020aa,Xu_PRL18,Wu_PRL18,Biao_PRL19,Rossi_FlatbandSC,FangXie_FlatbandSC,Wu_PRL20,Senthil_NearlyFlatBand,Di_PRL20,Li:2010aa,Trambly-de-Laissardiere:2010aa,PhysRevB.98.235158,XieMingPRL20,Stauber_Marginal_FL_inTBG,2020arXiv200810830S,2020arXiv200613963L,khalaf2020soft}.

There are eight flat bands arising from the combinations of degrees of freedom in the conduction bands of the component graphene layers \cite{Mele_TBG_Sym,Po_Symmetry_PRX,Zou_Symmetry_PRB,Oskar_Wannier,Liang_Wannier,Zhida_PRL19,Liang_CDW,Po_TBG_fragile,Carr_Wannier,Oscar_hiddensym,Oskar_PRL19,Zaletel_PRX20,Bernevig_tbg1,song2020tbg,bernevig2020tbg,lian2020tbg,bernevig2020tbg_5,xie2020tbg}. The states that comprise these bands may be labeled by a spin degree of freedom and two additional indices labeling the layer and sublattice of the component graphene sheets. Much of the novel physics of twisted bilayer graphene is believed to arise when the symmetries corresponding to these quantum numbers are spontaneously or explicitly broken. Interestingly, many of the broken symmetry states appear to have a topological character, revealed for example by anomalous Hall effects \cite{Serlin900,Sharpe605,Zaletel_TBG_AQH,Zaletel_PRX20,ZhaoLiu_TBG,Cecile_PRL20,Cecile_TBG_Flatband}, and at least at integer filling the topological character is believed to be inherent in the single-particle wavefunctions. An improved understanding of the single-particle wavefunctions is therefore important both for improved understandings of the observed and potentially observable topological phases and as a basis for theories of interaction effects in magic angle bilayer graphene.

Recently, Tarnopolsky, Kruchkov and Vishwanath \cite{Grisha_TBG} drew attention to a particular ``chiral''  model in which the interlayer tunneling Hamiltonian contains no terms in which an electron hops from one layer to the same sublattice on the other layer. They showed that in this case the eight weakly dispersing bands become exactly flat (dispersionless) at certain twist angles. They further constructed explicit expressions for the zero mode wavefunctions, and noticed that their solutions exhibited a holomorphic character reminiscent of the lowest Landau level quantum Hall physics \cite{Grisha_TBG,XiDai_PseudoLandaulevel}. This holomorphic character can give rise to a nontrivial topology of the flatbands, explaining the anomalous Hall effects.

In this paper, we study the zero mode wavefunctions of the chiral model \cite{Grisha_TBG} of twisted bilayer graphene in more detail. We identify an exact \ourinv symmetry of the chiral model and show how this symmetry implies that the flatband wavefunctions found by Tarpolsky, Kruchkov and Vishwanath can be written (up to a normalization factor) as:
\begin{equation}
\phi_{\bm k}(\bm r) = \left(\begin{matrix}i\mG(\bm r) \\\eta\mG(-\bm r)\end{matrix}\right)\times\Phi_{\bm k}(\bm r).\label{zeromodewavefunction}
\end{equation}
where $\Phi_{\bm k}$ is a quantum Hall wavefunction of the lowest Landau level, the function $\mG(\bm r)$ can be interpreted as a quantum Hall wavefunction in a magnetic field oppositely directed to that of $\Phi_{\bm k}$, and $\eta=\pm1$ is the inversion eigenvalue. The entire dependence on the crystal momentum $\bm k$ is carried by the quantum Hall wavefunction $\Phi_{\bm k}$, which exhibits one node at a $\bm k-$dependent position, while $\mG$, which is independent of $\bm k$, has a number of nodes that increases as the magic angle index increases, indicating a similarity between higher magic angles and higher Landau levels. This structure is revealed in FIG.~\ref{plotwfnormthreeangles}, which for the first three magic angles presents the norm of each component of $\phi_{\bm k}$ for the case where $\bm k$ is fixed at the \mr Dirac point $\bm K$ and implies a charge variation that can be detected by scanning tunneling spectroscopy.

We show that Eqn.~(\ref{zeromodewavefunction}) explains how the wavefunction $\phi_{\bm k}$ can simultaneously have the Abelian translation symmetry of the usual Bloch wavefunction and give rise to the anomalous Hall effect. Further, the quantum Hall anti-quantum Hall structure implies that the wavefunction nodes have a chirality and we find that nodes of a given chirality are concentrated in particular regions of the unit cell, implying intra-cell circulating currents that grow in magnitude as the magic angle increases. The increased density of nodes at higher magic angles will also affect the project of electron-electron interactions onto the flatbands in a manner similar to that occurring at higher Landau levels in the quantum Hall problem.

\begin{figure}[h]
    \centering
    \includegraphics[width=0.5\textwidth]{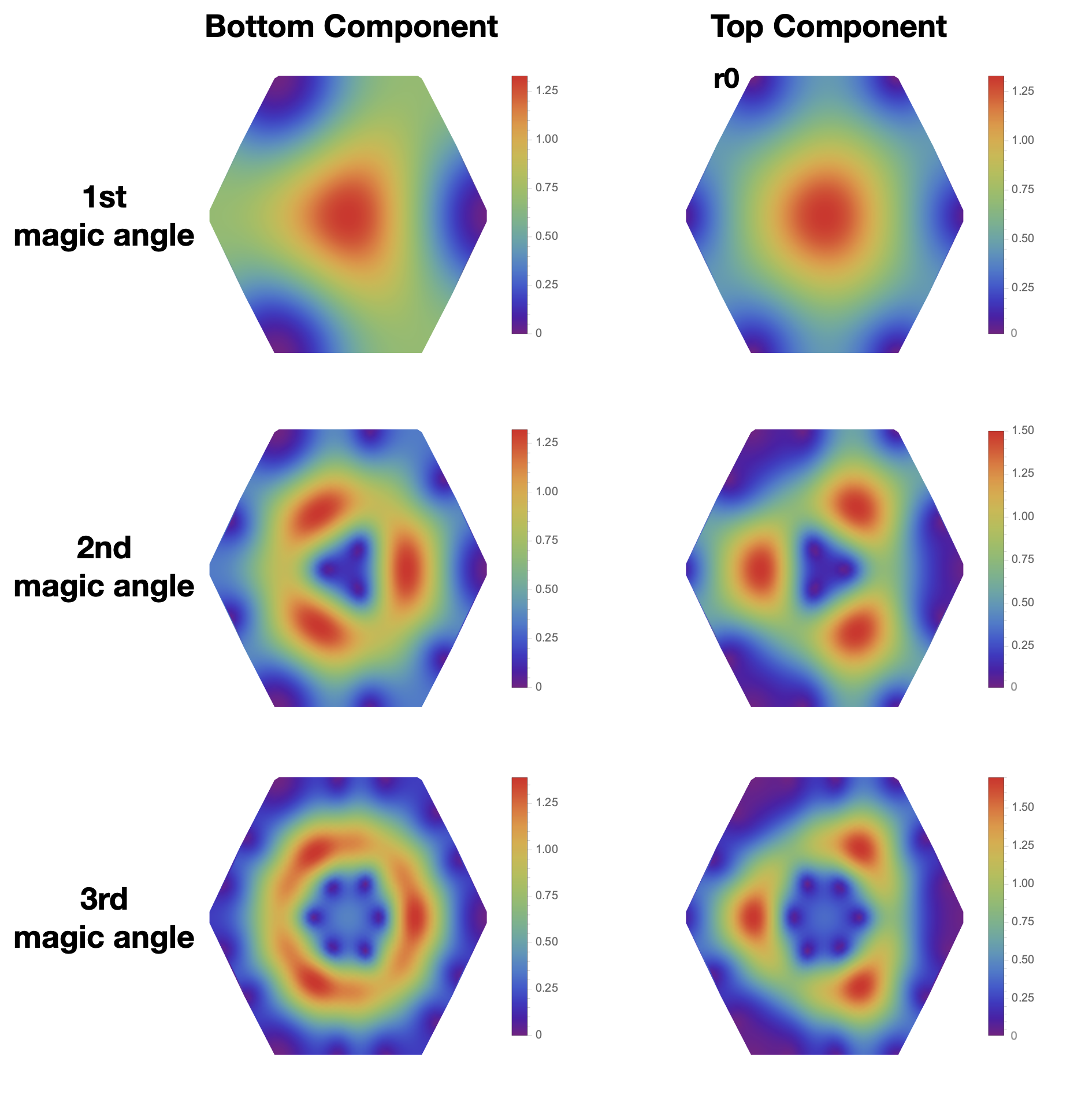}
    \caption{Norm of each component of the wavefunction Eqn.~(\ref{zeromodewavefunction}) at the \mr Dirac point $\bm K$, plotted at the first three magic angles (three rows). The upper left and lower right corners of the unit cells are the \emph{BA} ($\bm r_0$) and \emph{AB} ($-\bm r_0$) stacking points, as marked. The two columns correspond to the bottom and top components of the wavefunction. Each of them has clear symmetry and zero-structures. The wavefunctions at other \Blochk have a similar zero-structure, as explained in the text. The zeros are classified by their chirality {\it i.e.} whether the wavefunction's phase advances by $\pm 2\pi$ when encircling the zero once. Remarkably, the wavefunction associated with the $n_{th}$ magic angle has $3(n-1)$ zeros located at the unit cell center, all of which have the same chirality. We discuss the mathematical structure in Section~\ref{sec:spinor} and Section~\ref{sec:structure_zeros}, and implications for experimental observables in Section~\ref{Sec:ExperimentalObservation}.}\label{plotwfnormthreeangles}
\end{figure}

The paper is organized as follows. Section~\ref{sec:TBGreview} reviews the continuum model of twisted bilayer graphene and the chiral model defined from it, to establish the notation and approximations used here. Section~\ref{sec:symmetrysection} introduces our \ourinv symmetry and derives some of the properties that follow from it. Then in Section~\ref{sec:spinor}, we reexamine the derivation of the flatband wavefunctions and derive their spinor-structure. We then discuss the nodal structure of the flatband wavefunctions in Section~\ref{sec:structure_zeros}. In the last part of this work, Section~\ref{Sec:ExperimentalObservation}, we discuss how our findings can impact experimental observables. Section~\ref{sec:conclusion} is a summary and conclusion.

\section{Model Hamiltonians}\label{sec:TBGreview}
We start this section by reviewing the continuum model \cite{Santos,Santos2,Bistritzer12233} and the chiral model \cite{Grisha_TBG} of twisted bilayer graphene to establish the notation.

When two parallel graphene sheets (top, bottom) are stacked with any one of an infinite set of relative commensurate angles $\theta$, a \mr pattern forms, in which the combined system retains the basic hexagonal lattice structure of graphene, but with a much larger unit cell containing a number of carbon atoms $\sim \theta^{-2}$. The corresponding reciprocal space unit cell, which we refer to as the \mr Brillouin zone, is illustrated in FIG.~\ref{mBZ}.

As shown in FIG.~\ref{mBZ}, $\bm a_{i=1,2}$ indicate the two dimensional basis vectors of the \mr unit cell. The area of the \mr unit cell is $2\pi S$=$|\bm a_1\times\bm a_2|$. We denote the reciprocal lattice vectors as $\bm b^{i=1,2}$. Throughout this paper, we define the unit length by setting $\sqrt S$=$1$.

\begin{figure}[]
    \centering
    \includegraphics[width=0.35\textwidth]{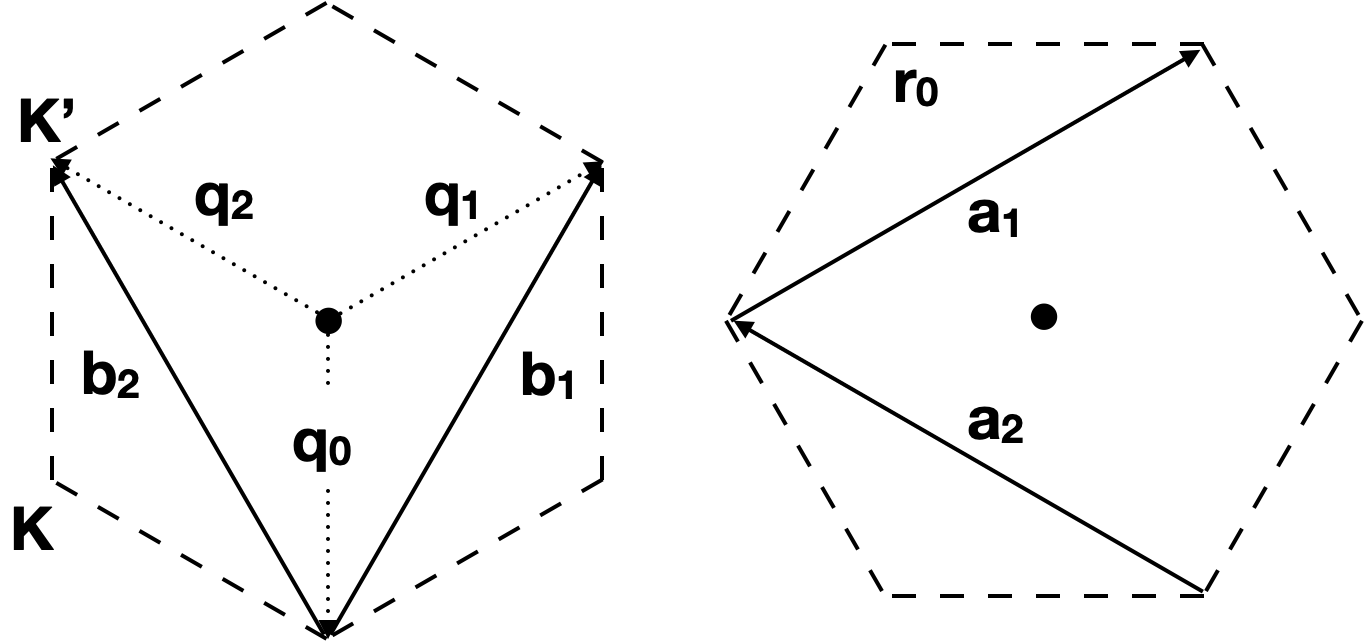}
    \caption{Left: \mr Brillouin zone. Right: real space \mr unit cell. Shown are the reciprocal lattice vectors $\bm b_{1,2}$, the \mr Dirac points $\bm K,\bm K'$, the real space lattice vectors $\bm a_{1,2}$ and the wavevectors $\bm q_{0,1,2}$, and the \emph{BA} stacking point $\bm r_0$. The \emph{AB} and \emph{AA} stacking points are located respectively at $-\bm r_0$ and the origin of the unit cell.}
    \label{mBZ}
\end{figure}

The fundamental single-particle Hamiltonian for twisted bilayer graphene consists of a standard single-layer graphene Hamiltonian for the top/bottom layer, $h_{G}(\bm r,\bm r^\prime)$, and an interlayer coupling $T(\bm r,\bm r^\prime)$ whose periodicity defines the \mr superlattice. Schematically the Hamiltonian operator is
 \begin{equation}
H_{TBLG}(\bm r,\bm r^\prime)=\left(\begin{array}{cc}h_G^b(\bm r,\bm r^\prime)& T(\bm r,\bm r^\prime) \\T^\dagger (\bm r,\bm r^\prime)& h_G^t(\bm r,\bm r^\prime)\end{array}\right).\label{HTBLG}
\end{equation}
where $t/b$ stands for the top/bottom graphene sheet.

It is generally agreed that, as proposed by Bistritzer and MacDonald \cite{Bistritzer12233}, the low energy properties of twisted bilayer graphene can be adequately described by a model with three key features. The first is a continuum description of the physics in each graphene sheet, obtained by linearizing the graphene Hamiltonian $h_G$ near the Dirac points (we denote the linearized Dirac Hamiltonian as $h_D$). The second is that the interlayer hopping only couples states near one Dirac point in one layer with states near the same graphene Dirac point in the other layer. This means that the relevant Hamiltonian is the product of two copies, one for each valley. A third simplification proposed by Bistritzer and MacDonald is that the interlayer hopping, in principle a function of $\bm r$ in one layer and $\bm r^\prime$ in the other becomes a function only of $\bm r$ with $\bm r^\prime$=$\bm r$. This is a coarse-graining approximation based on the notion that $T(\bm r)$ has a range of the order of the carbon-carbon distance so if the wavefunctions vary slowly on this scale we can ignore the detailed local structure. 

Following Bistritzer and MacDonald \cite{Bistritzer12233}, the effective continuum Hamiltonian of a single valley is,
\begin{equation}
H_{BM}=\int d^2\bm r \Psi_{BM}^\dagger(\bm r)\left(\begin{array}{cc}h_D^{b}\left(\frac{\theta}{2}\right)& T(\bm r) \\T^\dagger (\bm r)& h_D^{t}\left(-\frac{\theta}{2}\right)\end{array}\right)\Psi_{BM}(\bm r).\label{HTBG}
\end{equation}

A related Hamiltonian can be found for the opposite valley by acting with time reversal symmetry. The wavefunction $\Psi_{BM}(\bm r)$ is a four-component spinor, with the lower two components the two sublattices of the top layer, and the upper two the two sublattices of the bottom layer:
\begin{equation}
\Psi_{BM}(\bm r)=\left(\begin{array}{c}\psi^A_b(\bm r) \\\psi^B_b(\bm r) \\\psi^A_t(\bm r) \\\psi^B_t(\bm r)\end{array}\right).\label{macdonaldbasis}
\end{equation}

We have suppressed the spin index because the global $SU(2)$ spin invariance implies that the single-particle Hamiltonian is spin-diagonal. The continuum approximation to the Dirac Hamiltonian of a layer $\lambda=t,b$ is:
\begin{equation}
h_D^{\lambda}\left(\frac{\theta}{2}\right)=v_0\left(-i{\bm \nabla}-{\bm K}_+^{\lambda}\right)\cdot e^{-\frac{i\theta}{4}\sigma_z}\bm{\sigma}e^{\frac{i\theta}{4}\sigma_z}.\label{hktheta}
\end{equation}
where $\bm K_+^{t/b}$ is the graphene Dirac point $\bm K_+$ rotated by $\pm\theta/2$. As shown in FIG.~\ref{mBZ}, we define the \mr Dirac points as $\bm K$ = $\bm K_+^b-\bm K_+^{\Gamma}$, $\bm K'$ = $\bm K_+^t-\bm K_+^{\Gamma}$ where $\bm K_+^{\Gamma}$ is the \mr Gamma point labeled in graphene's reciprocal lattice coordinates. The interlayer tunneling potential $T(\bm{r})$ is constrained by the symmetries of a single valley: $\mC_3$, $\mM_y$ and $\mC_2\mT$, as discussed in Section~\ref{Sec:Symmetry}. In the Bistritzer-MacDonald model, the interlayer hopping is 
\begin{equation}
T(\bm r)=\sum_{j=0}^2T_je^{-i(\bm q_0-\bm q_j)\cdot\bm r}.\label{T}
\end{equation}
with $\phi$=$2\pi/3$, the $T_j$ is:
\begin{equation}
T_j= \omega_0 - \omega_1\cos(j\phi)\sigma_x + \omega_1\sin(j\phi)\sigma_y.\label{Tjdef}
\end{equation}

The \emph{chiral model} \cite{Grisha_TBG} is obtained by setting $\omega_0=0$ in Eqn.~(\ref{Tjdef}). The chiral model for a single valley is written in a different basis as $H_{BM}$ in Eqn.~(\ref{HTBG}):
\begin{equation}
H_{cBM}=\int d^2\bm r \Psi_c^\dagger (\bm r)\left(\begin{array}{cc}0& \mD(\bm r) \\ \mD^\dagger (\bm r)&0\\\end{array}\right)\Psi_c(\bm r).\label{chiral-form}
\end{equation}
where $\Psi_c(\bm r)$ is a four-component spinor whose upper two components ($\phi$) correspond to the $A$ sublattice of the bottom and top layer, and the lower two components ($\chi$) the $B$ sublattice of the bottom and top layer:
\begin{equation}
\Psi_c(\bm r)=\left(\begin{array}{c}\phi_b(\bm r) \\\phi_t(\bm r) \\\chi_b(\bm r) \\\chi_t(\bm r)\end{array}\right),\label{chiralbasis}
\end{equation}
where we have suppressed the \Blochk $\bm k$. The unitary transformation between the non-chiral basis Eqn.~(\ref{macdonaldbasis}) and the chiral basis Eqn.~(\ref{chiralbasis}) is:
\begin{equation}
\Psi_{c,\bm k}(\bm r) = e^{-i(\bm K_+^\Gamma+\tau_z\bm K)\cdot\bm r}e^{i\frac{\theta}{4}\tau_z\sigma_z}\Psi_{BM,\bm K_+^\Gamma+\bm k}(\bm r),\label{unitarytrans}
\end{equation}
where in Eqn.~(\ref{unitarytrans}), we have used $\bm\sigma$ and $\bm\tau$ for Pauli matrices acting on the sublattice and layer degrees of freedom respectively:
\beqn
\bm\sigma:~\text{sublattices};\quad\bm\tau:~\text{layers}.\nonumber
\eeqn

In Eqn.~(\ref{unitarytrans}), we have also shifted the center of the \Blochk of the chiral basis to the \mr Gamma point. The Bloch boundary condition of the chiral basis is:
\beqn
\Psi_{c,\bm k}(\bm r+\bm a) = e^{i(\bm k-\tau_z\bm K)\cdot\bm a}\Psi_{c,\bm k}(\bm r).\label{bc_chiralbasis}
\eeqn
where the details of Eqn.~(\ref{unitarytrans}) and Eqn.~(\ref{bc_chiralbasis}) can be found in Appendix~\ref{Utransformations}.

The operators $\mD^{\dag}(\bm r)$ and $\mD(\bm r)$ in Eqn.~(\ref{chiral-form}) are:
\beqn
\mD^{\dag}(\bm r) &=& \sqrt{2}\left(\begin{matrix}-i\bar\partial & \alpha U_{\phi}(\bm r)\\\alpha U_{\phi}(-\bm r) & -i\bar\partial\end{matrix}\right),\nonumber\\
\mD(\bm r) &=& \sqrt{2}\left(\begin{matrix}-i\partial & \alpha U_{-\phi}(\bm r) \\\alpha U_{-\phi}(-\bm r) & -i\partial\end{matrix}\right).\label{AshvinHamiltonian}
\eeqn
where $U_{\phi}(\bm r)$ is:
\begin{equation}
U_{\phi}(\bm r) = e^{-i\bm q_0\cdot\bm r} + e^{i\phi}e^{-i\bm q_1\cdot\bm r} + e^{-i\phi}e^{-i\bm q_2\cdot\bm r}.\label{Uphi}
\end{equation}

As usual, we have defined $z=(x+iy)/\sqrt2$, $\partial=(\partial_x-i\partial_y)/\sqrt2$. The parameter $\alpha$ is determined by the twisted angle: $\alpha=(3w_1a_0)/(8\sqrt2\pi v_0\sin\frac{\theta}{2})$ where $v_0$ is the graphene's Fermi velocity and $a_0$ is the graphene's lattice constant. The vectors $\bm q_{0,1,2}$ are specified in FIG.~\ref{mBZ}.

The chiral Hamiltonian anti-commutes with the chiral matrix $\sigma_z$. As a consequence, the single-particle spectrum is particle-hole symmetric. In the next section, we review symmetries of twisted bilayer graphene, and introduce the \ourinv symmetry.

\section{Intra-Valley Inversion Symmetry}\label{sec:symmetrysection}
In this section, we start by discussing the symmetries of twisted bilayer graphene with an emphasis on how $\mC_2\mT$ symmetry constrains the tunneling terms. In Section~\ref{sec:exact_inversion_sym}, we introduce the exact \ourinv symmetry of the chiral model, and derive some properties that follow from it.

\subsection{Symmetry constraint on tunneling terms}\label{Sec:Symmetry}
The symmetries of twisted bilayer graphene play crucial roles in determining the single and many particle properties \cite{Mele_TBG_Sym,Po_Symmetry_PRX,Zou_Symmetry_PRB,Oskar_Wannier,Liang_Wannier,Zhida_PRL19,Liang_CDW,Po_TBG_fragile,Carr_Wannier,Oscar_hiddensym,Oskar_PRL19,Zaletel_PRX20,Bernevig_tbg1,song2020tbg,bernevig2020tbg,lian2020tbg,bernevig2020tbg_5,xie2020tbg}. In this section, we review these symmetries, with an emphasis on how symmetries constrain the low energy continuum Hamiltonian.

The ``crystal symmetries'' of twisted bilayer graphene are generated by the \mr translation symmetry, $\mC_6$ rotational rotation, and mirror symmetry $\mM_y$. Time reversal symmetry, $\mT$, is also present. In addition, in the continuum model the charge conservation of each valley, {\it i.e.} $U(1)$ valley symmetry, is assumed. The symmetries that keep each valley invariant ($\mC_2\mT$, $\mC_3$ and $\mM_y$) constrain the single valley Hamiltonian in Eqn.~(\ref{HTBG}) and Eqn.~(\ref{chiral-form}). Here, the important constraint for us is that $\mC_2\mT$ symmetry requires the tunneling term in Eqn.~(\ref{HTBG}) satisfy (proof in Appendix~\ref{sec:append_sym}):
\beqn
T(\bm r) = \sigma_x T^*(-\bm r)\sigma_x,\label{C2T_T}
\eeqn
where, as in the previous section, $\sigma_x$ acts on sublattice space. In the chiral basis, this means that the off-diagonal elements of $\mD$ (and $\mD^\dagger$) are related by $\bm r \leftrightarrow -\bm r$, as we shown in Eqn.~(\ref{AshvinHamiltonian}).

\subsection{Exact \ourinv symmetry of the chiral model}\label{sec:exact_inversion_sym}
Here we show that the chiral model enjoys an exact \ourinv symmetry as constrained by $\mC_2\mT$ symmetry and the linearized Dirac fermion. We then discuss properties that follow from it, including the symmetries of the spectrum and single-particle states. In the end, we show a numerically observed alternating pattern of magic angle inversion parities.

\begin{lemma}
\label{lemmaDsigmay}
The zero-mode operator satisfies,
\begin{equation}
\tau_y\mD^{\dag}(\bm r)\tau_y = -\mD^{\dag}(-\bm r).\label{involution}
\end{equation}
\end{lemma}

The calculation follows from the $\mC_2\mT$ constraint in Eqn.~(\ref{C2T_T}) and the definition of $\mD^{\dag}$ in Eqn.~(\ref{AshvinHamiltonian}). For Lemma~\ref{lemmaDsigmay} to hold, we need the off-diagonal elements of the $\mD(\bm r)$ operator to be related by $\bm r\leftrightarrow-\bm r$. In other words, in the chiral basis (Eqn.~(\ref{chiralbasis})), the interlayer tunneling potential from top to bottom layer is identical to that from bottom to top layer with spatial inversion. As we discussed in Section~\ref{Sec:Symmetry}, this is guaranteed by the $\mC_2\mT$ symmetry.

We now define the \ourinv symmetry.

\begin{theorem}\label{theoreminv}
The chiral model of twisted bilayer graphene has an exact \ourinv symmetry, whose operator is,
\begin{equation}
\mI \equiv \sigma_z\tau_y.
\end{equation}
such that,
\begin{equation}
\mI\mH(\bm r)\mI^\dag = \mH(-\bm r).\label{Hpseudoinvsym}
\end{equation}
Again, $\bm\sigma$ and $\bm\tau$ are Pauli matrices acting on the sublattice and layer degrees of freedom respectively.
\begin{proof}
It is straightforward to prove by using Lemma~\ref{lemmaDsigmay}:
\beqn
\mI\mH(\bm r)\mI^{\dag} &=& \left(\begin{matrix}\tau_y&0\\0&-\tau_y\end{matrix}\right)\left(\begin{matrix}0&\mD(\bm r)\\\mD^{\dag}(\bm r)&0\end{matrix}\right)\left(\begin{matrix}\tau_y&0\\0&-\tau_y\end{matrix}\right),\nonumber\\
&=& -\left(\begin{matrix}0&\tau_y\mD(\bm r)\tau_y\\\tau_y\mD^{\dag}(\bm r)\tau_y&0\end{matrix}\right) = \mH(-\bm r).\label{Hpseudoinv}
\eeqn
\end{proof}
\end{theorem}

We call Eqn.~(\ref{Hpseudoinvsym}) the \ourinv symmetry in order to distinguish it from the crystalline 2D inversion symmetry, $\mC_2$. Since the $\mC_2$ symmetry mixes valleys of twisted bilayer graphene, it is not a symmetry of the single-valley continuum models in Eqn.~(\ref{HTBG}) and Eqn.~(\ref{chiral-form}). In contrast, the \ourinv maps $\bm k$ to $-\bm k$ \emph{within} the \mr Brillouin zone and thus does not mix valleys. As shown in Eqn.~(\ref{Hpseudoinvsym}), the \ourinv symmetry is an exact symmetry for the single valley chiral model Eqn.~(\ref{chiral-form}).

We emphasize that the \ourinv symmetry requires \emph{no} extra assumptions beyond the chiral model in Eqn.~(\ref{chiral-form}). The only requirement is the crystal symmetry $\mC_2\mT$ and the linearized Dirac fermion, which are already present in the chiral model in Eqn.~(\ref{chiral-form}).

The $\tau_y$ operator of Eqn.~(\ref{involution}) has appeared in recent literature. In Ref.~(\onlinecite{becker2020mathematics}) it is referred to as the \emph{involution operator}. It also appeared as Eqn.~(S15) in the supplementary material of Ref.~(\onlinecite{Zaletel_PRX20}). In Ref.~(\onlinecite{Zhida_PRL19}) and a very recent paper Ref.~(\onlinecite{song2020tbg}), by the same authors, a similar operator $i\tau_y$ plus $\bm r\leftrightarrow-\bm r$ is termed the \emph{unitary particle-hole} operator. This is different than our \ourinv symmetry: our $\tau_y$ operates on the chiral basis in Eqn.~(\ref{chiralbasis}), while the ``unitary particle-hole'' acts on the non-chiral basis in Eqn.~(\ref{macdonaldbasis}). Since the unitary transformation between these two bases in Eqn.~(\ref{unitarytrans}) does not commute with $\tau_y$, these two symmetries are distinct. It is also important to emphasize that our \ourinv is an \emph{exact symmetry of the chiral model}, while the unitary particle-hole symmetry is an approximate symmetry for both the continuum model in Eqn.~(\ref{HTBG}) and the chiral model in Eqn.~(\ref{chiral-form}), according to Ref.~(\onlinecite{Zhida_PRL19}) and Ref.~(\onlinecite{song2020tbg}).

Many interesting facts follow from the \ourinv symmetry, as we describe here and in the next section.

\begin{corollary}
\label{inversion}
At all twist angles, the single particle spectrum of the chiral model is not only particle-hole symmetric, but also inversion symmetric.
\end{corollary}

This follows directly from Theorem~\ref{theoreminv}. Denote the sublattice A/B wavefunctions as
\beqn
\Psi_{\bm k} = \left(\begin{matrix}\phi_{\bm k}\\\chi_{\bm k}\end{matrix}\right),\nonumber
\eeqn
where each of $\phi_{\bm k}$ and $\chi_{\bm k}$ is a two-component spinor representing the bottom and top layer's degrees of freedom. If we know $\Psi_{\bm k}$ as an eigenstate of energy $E$ at \Blochk $\bm k$, then $\mI\Psi_{\bm k}(-\bm r)$ is the eigenstate of the same energy but with an opposite \Blochkk:
\beqn
\mH(\bm r)\mI\Psi_{\bm k}(-\bm r) = \mI\mH(-\bm r)\Psi_{\bm k}(-\bm r) = E\mI\Psi_{\bm k}(-\bm r).\nonumber
\eeqn

We have thus proved the spectrum inversion symmetry by explicitly constructing eigenstates of the same energy and opposite \Blochkk. This construction in fact also illustrates a spinor structure of the eigenstates.

\begin{theorem}
\label{psiminusk}
At all twist angles for any \Blochk $\bm k$, there exists a phase $\zeta_{\bm k}$, such that,
\beqn\phi_{\bm k}(\bm r) &=&+e^{i\zeta_{\bm k}}\tau_y\phi_{-\bm k}(-\bm r),\label{remarkinvk}\\
\chi_{\bm k}(\bm r) &=&-e^{i\zeta_{\bm k}}\tau_y\chi_{-\bm k}(-\bm r),\nonumber\\
\zeta_{\bm k} &=& -\zeta_{-\bm k}.\nonumber
\eeqn
\end{theorem}
\begin{proof}
Below Corollary~\ref{inversion}, we explicitly constructed the eigenstate of opposite \Blochkk. At non-degenerate $\bm k$, our constructed wavefunction must be proportional to the wavefunction at $-\bm k$ up to a $U(1)$ phase,
\begin{equation}
\Psi_{-\bm k}(\bm r) = e^{i\zeta_{\bm k}}\sigma_z\tau_y\Psi_{\bm k}(-\bm r),\nonumber
\end{equation}
from which Eqn.~(\ref{remarkinvk}) follows immediately. The fact that the phase $\zeta_{\bm k}$ is anti-symmetric is seen by applying Eqn.~(\ref{remarkinvk}) twice. For degenerate zero modes, one can label them by the chiral eigenvalue and find the same conclusion.
\end{proof}

Theorem~\ref{psiminusk} can be regarded as a gauge fixing condition. One can perform gauge transformations,
\begin{equation}
\phi_{\bm k} \rightarrow e^{i\zeta'_{\bm k}}\phi_{\bm k}.
\end{equation}
to tune the $\zeta_{\bm k}$ field:
\begin{equation}
\zeta_{\pm\bm k} \rightarrow \zeta_{\pm\bm k} \mp (\zeta'_{\bm k} - \zeta'_{-\bm k}).
\end{equation}

The only obstruction of such tuning is at inversion symmetric points where $\bm k_{inv} = -\bm k_{inv}$ modulo reciprocal lattice vectors: there $\zeta'_{\bm k}$ and $\zeta'_{-\bm k}$ cancel, and $\zeta_{\bm k_{inv}}$ is either $0$ or $\pi$. In practice, the \ourinv eigenvalue can be read off from the transformation property of the \mr Gamma point wavefunction (or from wavefunctions at other $\bm k_{inv}$):
\begin{equation}
\phi_{\bm k=\bm 0}(\bm r) = \eta\tau_y\phi_{\bm k = \bm 0}(-\bm r),\quad\eta=\pm1.\label{readoffinv}
\end{equation}

We numerically observed (for the first three magic angles) that there is a coincidence between the zero mode's \ourinv eigenvalue and the parity of magic angle: we found $\eta$=$+1$ for the 1\ts{st}, 3\ts{rd} magic angles, while for the 2\ts{nd} magic angle $\eta$=$-1$. In FIG.~\ref{evolution_1_2}, we monitored the evolution of the low lying eigenstates at the Gamma point from the 1\ts{st} to the 2\ts{nd} magic angle (we plot non-negative energies only since the full spectrum has particle-hole symmetry). Eigenstates at the Gamma point are either singlet or doublet, as they are one and two dimensional irreducible representation of the symmetry group (generated by $\mC_3$ and $\mM_y$) \cite{Oskar_Wannier}. An inversion eigenvalue transition is clearly visible in FIG.~\ref{evolution_1_2}.

\begin{figure}[h]
    \centering
    \includegraphics[width=0.4\textwidth]{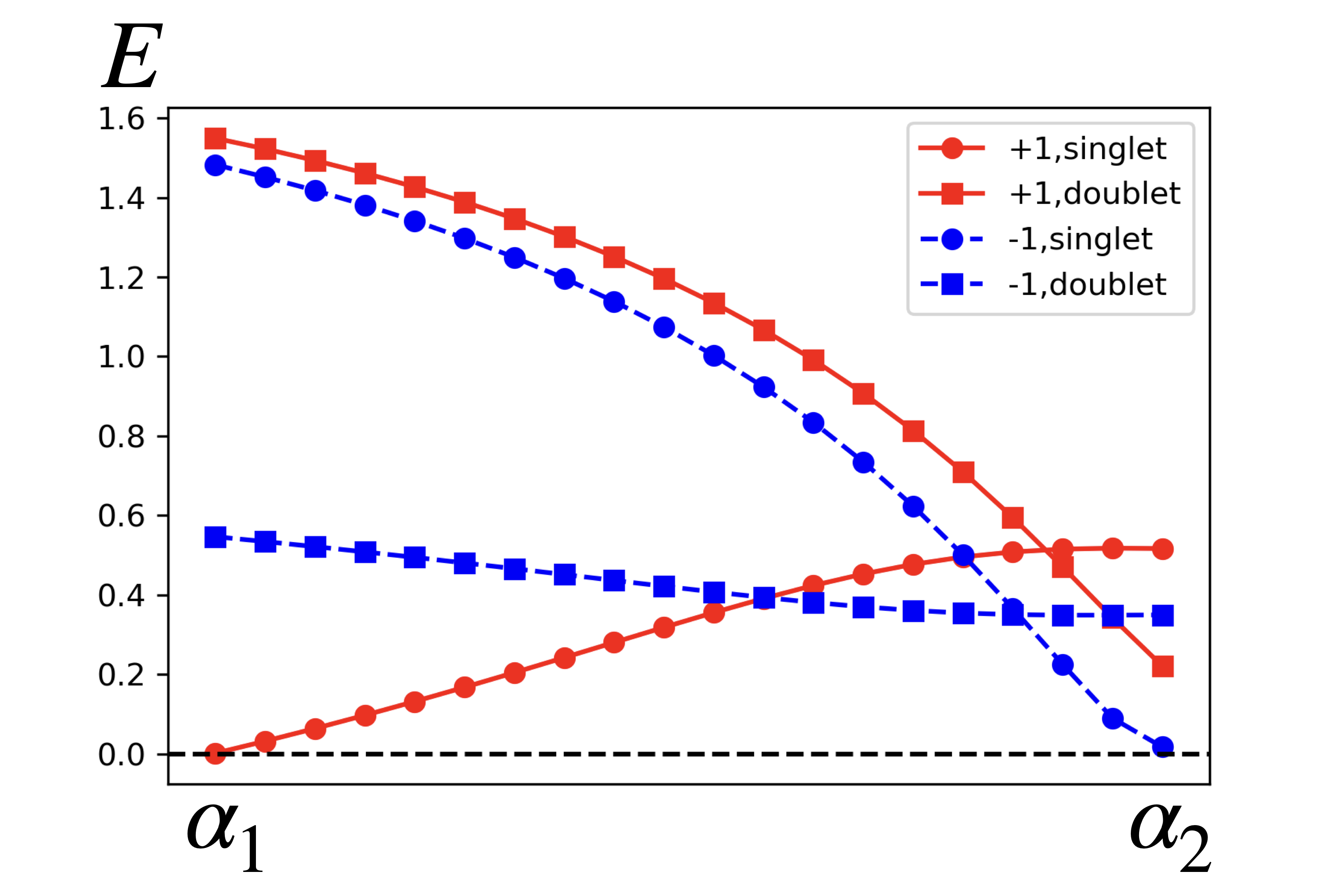}
    \caption{Evolution at the \mr Gamma point of low lying non-negative energy states from the first magic angle $\alpha_1$ to the second magic angle $\alpha_2$. A singlet (doublet) state is represented as dots (squares). The inversion symmetric (anti-symmetric) state is represented by red solid (blue dashed) lines. The black horizontal line indicates zero energy. Due to the particle-hole symmetry of the chiral model, the evolution of negative energy states is obtained by reflecting the figure.}\label{evolution_1_2}
\end{figure}

We hypothesize that the alternating parity of magic angle zero mode wavefunctions is a generic feature and will hold for all magic angles: that is, the inversion eigenvalue of $n_{th}$ magic angle flatband wavefunction is $-(-1)^n$.

\section{Zero Mode Wavefunctions and the Spinor Structure}\label{sec:spinor}
In this section, we reexamine the zero mode solution of Ref.~(\onlinecite{Grisha_TBG}), derive the spinor structure of the zero mode wavefunction as shown in Eqn.~(\ref{zeromodewavefunction}), and show its intricate relation to quantum Hall physics.

\subsection{Chiral model and the zero modes}\label{Sec:ZeroModesChiralModel}
Following Ref.~(\onlinecite{Grisha_TBG}), we show that the chiral Hamiltonian in Eqn.~(\ref{chiral-form}) has two zero modes. The eigenvectors of these two zero modes must satisfy:
\begin{equation}
0=\mD(\bm r)\left(\begin{array}{c}\chi_{b,\bm k}(\bm r) \\\chi_{t,\bm k}(\bm r)\end{array}\right);\quad0=\mD^\dagger(\bm r)\left(\begin{array}{c}\phi_{b,\bm k}(\bm r) \\\phi_{t,\bm k}(\bm r)\end{array}\right).\label{zeromode2component}
\end{equation}

In Ref.~(\onlinecite{Grisha_TBG}), the authors found and proved that, for a discrete series of values of $\alpha$ (corresponding to magic angles), the chiral model admits exact flatbands. At the crux of their analysis is the fact that at magic angles, both components of the \mr Dirac $\bm K$ point wavefunction $\phi_{\bm K}=(\phi_{b,\bm K},\phi_{t,\bm K})^T$ vanish at a common point:
\begin{equation}
\bm r_0=\frac{1}{3}(\bm a_1+2\bm a_2),\label{def_r0}
\end{equation}
the \emph{BA} stacking point, which permits an explicit construction of the zero mode wavefunctions. Here we review several key steps (Theorem~\ref{r0-zero} to Theorem~\ref{cTBGwavefunction}) of Ref.~(\onlinecite{Grisha_TBG}) in deriving the zero mode wavefunctions. We refer the readers to Ref.~(\onlinecite{Grisha_TBG}) for more details.

A crucial step in deriving the zero mode solutions in Ref.~(\onlinecite{Grisha_TBG}) is Theorem~\ref{r0-zero}, which follows from the translation and $\mC_3$ rotation symmetry:
\begin{theorem}\label{r0-zero}
For all twisted angles, $\phi_{\bm K,t}(\pm\bm r_0)$=$0$ and $\chi_{\bm K,t}(\pm\bm r_0)$=$0$.
\end{theorem}

\begin{theorem}
\label{fermi-velocity}
The Fermi velocity defined by:
\beqn
v_F(\alpha) \equiv \sum_{l=t/b}\phi_{l,\bm K}(\bm r)\phi_{l,\bm K}(-\bm r),
\eeqn
is independent of $\bm r$.
\end{theorem}
\begin{proof}
It is straightforward to find $v_F(\alpha)$ is holomorphic, {\it i.e.} $\bar\partial v_F(\alpha)=0$, by using the zero mode equations that $\phi_{l,\bm K}$ satisfy. Then, $v_{F}(\alpha)$ must be a constant since it is also cell-periodic.
\end{proof}

At magic angles (where the low lying two bands become dispersionless), the Fermi velocity goes to zero. Since the top component of $\phi_{\bm K}$ vanishes for all twist angles at $\pm\bm r_0$, it follows from the vanishing Fermi velocity that the bottom component must at least have one common zero with the top component, at either $+\bm r_0$ or $-\bm r_0$. In fact, the exact flatband condition coincides with the condition that two components of the wavefunction have a common zero, as pointed out in Ref.~(\onlinecite{Grisha_TBG}), where the authors proved it by explicitly constructing the zero mode wavefunctions.

\begin{theorem}
\label{cTBGwavefunction}
The magic angle zero mode wavefunctions take the following form \cite{Grisha_TBG,Grisha_TBG2} (up to a normalization factor),
\beqn
\phi_{\bm k}(\bm r) &=& \phi_{\bm K}(\bm r)F_{\bm k}(z),\nonumber\\
\quad F_{\bm k}(z) &=& e^{z_k^*(z-\frac{1}{2}z_k)}\frac{\sigma(z-z_k)}{\sigma(z-z_0)},\label{cTBGwf}
\eeqn
where $z_0$ and $z_k$ are the complex coordinates of $\bm r_0$, the \emph{BA} stacking point defined in Eqn.~(\ref{def_r0}), and:
\begin{equation}
\bm r_{\bm k}^a = \bm r_0^a + \epsilon^{ab}(\bm k-\bm K)_b.\label{defrk}
\end{equation}
The complex coordinate for a vector $\bm r$ is defined as usual:
\begin{equation}
\bm r \rightarrow z \equiv \frac{r_x + i r_y}{\sqrt2}.
\end{equation}
\end{theorem}

Note here we have written the zero mode wavefunction in terms of the ``modified Weierstrass sigma'' function $\sigma(z)$, which is slightly different from Ref.~(\onlinecite{Grisha_TBG}), where the authors used Jacobi theta functions. It has been shown \cite{haldanemodularinv,Jie_MonteCarlo,scottjiehaldane,Jie_Dirac} that both the sigma function and theta function can be used to define the quantum Hall states, and the advantage of the former is modular invariance. The Weierstrass sigma function satisfies a similar quasi-periodic boundary condition as the Jacobi theta function:
\beqn
\sigma(z+a_i)=-e^{a_i^*(z+\frac{1}{2}a_i)}\sigma(z),\label{periodicity_sigma}
\eeqn
where $a_{i=1,2}$ are the complex coordinates of the primitive lattice vectors $\bm a_{1,2}$ shown in FIG.~(\ref{mBZ}). The quantum Hall wavefunction and the modified Weierstrass sigma function $\sigma(z)$ are reviewed in detail in Appendix~\ref{sec:QH_wavefunction}. Note that the factor $\exp(-\frac{1}{2}|z_k|^2)$ in Eqn.~(\ref{cTBGwf}) is needed to ensure that the normalization is periodic in $\bm k$.

The presence of the quasi-periodic elliptic function in the zero mode solution is reminiscent of the lowest Landau level physics on torus \cite{haldanetorus1,XiDai_PseudoLandaulevel}. We find it conceptually and practically advantageous to rewrite Eqn.~(\ref{cTBGwf}) in the following form, as a product of a quantum Hall wavefunction and a quasi-periodic spinor wavefunction:
\begin{equation}
\phi_{\bm k}(\bm r) = \left(\begin{matrix}\mathcal{G}_1(\bm r)\\\mathcal{G}_2(\bm r)\end{matrix}\right)\times\Phi_{\bm k}(\bm r).
\end{equation}
where $\mathcal{G}_{1/2}(\bm r) \equiv \phi_{\bm K,b/t}(\bm r)/\left(\sigma(z-z_0)e^{-\frac{1}{2}|z|^2}\right)$ and the quantum Hall wavefunction $\Phi_{\bm k}$ is,
\begin{equation}
\Phi_{\bm k}(\bm r) = e^{z_k^*z}\sigma(z-z_k)e^{-\frac12|z_k|^2}e^{-\frac12|z|^2},\label{quantumhallwf}
\end{equation}
whose boundary condition can be found in Eqn.~(\ref{QHbc}) in Appendix~\ref{sec:QH_wavefunction}.

Reformulating the zero mode wavefunction in this way makes the subsequent discussions in Section~\ref{sec:structure_zeros} more clear.

\subsection{Spinor structure of zero mode wavefunctions}\label{sec:flatbandzerostructure}
The \ourinv implies that the two components of the (magic angle) zero mode wavefunctions are not independent.
\begin{theorem}
\label{spinor}
The zero mode wavefunction can be written as Eqn.~(\ref{zeromodewavefunction}), which we copy below,
\begin{equation}
\phi_{\bm k}(\bm r) = \left(\begin{matrix}i\mG(\bm r) \\\eta\mG(-\bm r)\end{matrix}\right)\times\Phi_{\bm k}(\bm r),\nonumber
\end{equation}
where $\eta=\pm1$ is the \ourinv eigenvalue from Eqn.~(\ref{readoffinv}) and $\Phi_{\bm k}(\bm r)$ is the quantum Hall wavefunction Eqn.~(\ref{quantumhallwf}).
\end{theorem}
\begin{proof}
We start with the ansatz:
\begin{equation}
\phi_{\bm k}(\bm r) = \left(\begin{matrix}\mG_1(\bm r)\\\mG_2(\bm r)\end{matrix}\right)\Phi_{\bm k}(\bm r).\label{psiansatz}
\end{equation}
Applying Theorem~\ref{psiminusk} yields:
\beqn
\phi_{\bm k}(\bm r) &=& e^{i\zeta_{\bm k}}\left(\begin{matrix}i\mG_2(-\bm r)\\-i\mG_1(-\bm r)\end{matrix}\right)\Phi_{\bm k}(\bm r),\label{psiansatz2}
\eeqn
where we have used the inversion property of the quantum Hall wavefunction $\Phi_{-\bm k}(-\bm r) = -\Phi_{\bm k}(\bm r)$ (derived in Appendix~\ref{sec:qhwf}). Equating Eqn.~(\ref{psiansatz}) and Eqn.~(\ref{psiansatz2}) yields,
\begin{equation}
\phi_{\bm k}(\bm r) = \left(\begin{matrix}i\mG(\bm r) \\ e^{i\zeta_{\bm k}}\mG(-\bm r)\end{matrix}\right)\Phi_{\bm k}(\bm r),\quad e^{i\zeta_{\bm k}} = \pm 1.
\end{equation}
where we defined $\mG(\bm r) \equiv -i\mG_1(\bm r)$.
\end{proof}

The boundary condition of $\mG(\bm r)$ is derived in Eqn.~(\ref{boundary-condition-G}) of Appendix~\ref{sec:QH_wavefunction}.

To conclude, following the \ourinv symmetry, we have derived the spinor structure of the zero mode wavefunction as shown in Eqn.~(\ref{zeromodewavefunction}) and have demonstrated explicitly its connection to the lowest Landau level wavefunctions. The $\eta$ in Eqn.~(\ref{zeromodewavefunction}) is the \ourinv eigenvalue, which can be read off from Eqn.~(\ref{readoffinv}).

\section{Nodal Structure}\label{sec:structure_zeros}
In the previous sections, we described an \ourinv symmetry of the chiral model, which led to the discovery of the spinor structure of the zero mode wavefunctions. There we factorized the wavefunction into a quantum Hall wavefunction and a pre-factor $\mG(\bm r)$.

However, so far the physical interpretation of the function $\mG(\bm r)$ remains mysterious, as does the structure of zeros in FIG.~\ref{plotwfnormthreeangles}. One hint is that the zero modes must be Bloch functions that transform under the usual translation group, while quantum Hall states transform under the magnetic translation group. Hence $\mG(\bm r)$ must also be quasi-periodic to ``cancel'' the magnetic translation effects of the quantum Hall wavefunction.

In this section, we resolve this puzzle by demonstrating mathematically and numerically that $\mG(\bm r)$ can be regarded as an anti-quantum Hall wavefunction at a certain Landau level, {\it i.e.} a quantum Hall state in a magnetic field oppositely directed to that of $\Phi_{\bm k}$, with the order of the magic angle serving the role of the Landau level index. In this way, the zero mode wavefunction is a product of a quantum Hall and an anti-quantum Hall state, whose net magnetic fluxes passing through the \mr unit cell cancel, allowing the whole wavefunction to be a usual Bloch function. We then discuss the zeros in more detail. In the next section, we discuss its experimental implications.

\subsection{Analytical expansion of $\mG(\bm r)$}
To demonstrate the anti-quantum Hall nature of $\mG(\bm r)$, we will start by showing that the leading order expansion near $\bm r_0$ is anti-holomorphic:
\beqn
\mG(\bm r_0 + \bm r) \sim \bar z.\label{vanishingG_r0}
\eeqn

Hence we can peel off an anti-quantum Hall wavefunction from the zero mode wavefunction and rewrite its components as Eqn.~(\ref{rhob}) and Eqn.~(\ref{rhot}).

Since $\mG(\bm r)$ is independent of \Blochkk, without loss of generality we can consider the \mr $\bm K$ point sublattice A wavefunction $\phi_{\bm K}$=$(\phi^b_{\bm K},\phi^t_{\bm K})^T$ to analyze. Its zero mode equation $\mD^{\dag}(\bm r)\phi(\bm r)$=0 implies a relation between the top and bottom components $\phi^t_{\bm K}$=$i\bar\partial\phi^b_{\bm K}/(\alpha U_{\phi})$. Theorem~\ref{r0-zero} tells us that $\phi^t_{\bm K}$ must have zeros \cite{Grisha_TBG} at $\pm\bm r_0$. From the form of the zero mode wavefunction Eqn.~(\ref{zeromodewavefunction}), we know that the $+\bm r_0$ and $-\bm r_0$ zeros of $\phi^t_{\bm K}$ come, respectively, from the quantum Hall part $\Phi_{\bm k}$ and $\mG(-\bm r)$. Therefore, near $\bm r_0$, $\phi^t_{\bm K}$ must vanish holomorphically:
\beqn
i\frac{\bar\partial\phi^b_{\bm K}(\bm r_0+\bm r)}{\alpha U_{\phi}(\bm r_0+\bm r)} \sim z.\label{expansionr0}
\eeqn
Then, by using $U_{\phi}(\bm r_0)=3$ and the $\mC_3$ symmetry, one can see that $\phi^b_{\bm K}$ must have a second order zero at $\bm r_0$, vanishing as: $\phi^b_{\bm K}(\bm r_0+\bm r)$$\sim$$z\bar z$. Again according to Eqn.~(\ref{zeromodewavefunction}), $z$ and $\bar z$ of the bottom component $\phi^b_{\bm K}$ come, respectively, from the quantum Hall wavefunction and $\mG(\bm r)$. We hence justified Eqn.~(\ref{vanishingG_r0}).

\subsection{Zero mode wavefunction revisited}\label{revisit_zeromode_eqn}
The vanishing behavior of $\mG(\bm r)$ near $\bm r_0$ shows it is possible to factorize out an anti-quantum Hall wavefunction (a quantum Hall state in a magnetic field oppositely directed to that of $\Phi_{\bm k}$, which we denote as $\bar\Phi_{\bm k}$$\equiv$$(\Phi_{\bm k})^*$) from it without encountering singularities. The \Blochk $\bm k$ of $\bar\Phi_{\bm k}$ is determined by the Bloch translation symmetry of the whole wavefunction. After some algebra, we end up with the final expression:
\begin{equation}
\phi^b_{\bm k}(\bm r) = i\rho( \bm r)\times\bar\Phi_{\bm K }(\bm r)\Phi_{\bm k}(\bm r),\label{rhob}
\end{equation}
where we introduced a function $\rho(\bm r)$ which must be cell-periodic due to the cancellation of the non-periodic parts from $\Phi_{\bm k}$ and $\bar\Phi_{\bm k}$:
\begin{equation}
\rho(\bm r)\equiv\mG(\bm r)/\bar\Phi_{\bm K}(\bm r).\label{defrho}
\end{equation}

The top layer wavefunction is obtained easily by the \ourinv symmetry:
\begin{equation}
\phi^t_{\bm k}(\bm r) = -\eta\rho(-\bm r) \times \bar\Phi_{\bm K'}(\bm r)\Phi_{\bm k}(\bm r).\label{rhot}
\end{equation}

The $\Phi_{\bm k}$ and $\bar\Phi_{\bm K}$ of Eqn.~(\ref{rhob}) carry opposite magnetic fields that cancel with each other, leaving $\phi^{b/t}_{\bm k}$ as a Bloch state. Since the crystal momentum ($\bm k$) dependence, and hence response to an external electric field, is only from the $\Phi_{\bm k}$ piece, the wavefunction $\phi_{\bm k}$ should have the same topological character as the lowest Landau level wavefunction, according to Laughlin's gauge invariance argument \cite{Laughlin_Chargepump}.

To see how this argument applies to our case more explicitly, imagine we apply a time-independent and spatially uniform in-plane external electric field $\bm E$ across the twisted bilayer graphene sample. The \Blochk of the electron couples to $\bm E$ through minimal couping, and consequently changes linearly with time: $\delta\bm k\sim\bm Et$. According to Eqn.~(\ref{defrk}), we know that the zero of $\Phi_{\bm k}$ is locked to $\bm k$, and moves in the direction perpendicular to $\bm E$. Since zero corresponds to a charge minimum, we conclude that a unit of charge is adiabatically pumped in a direction perpendicular to $\bm E$ during a unit of time. This demonstrates that the zero mode wavefunction Eqn.~(\ref{zeromodewavefunction}), as a product of a quantum Hall wavefunction and an anti-quantum Hall wavefunction, is indeed a Bloch function which carries Chern number $\mC$=$1$. So far, we discussed the sublattice-A polarized flatband wavefunction $\phi_{\bm k}(\bm r)$. The other degenerate flatband $\chi_{\bm k}(\bm r)$ is sublattice-B polarized and has Chern number $\mC$=$-1$ since these two flatbands are related by the $\mC_2\mT$ symmetry before considering the hexagonal boron nitride substrate. With the hexagonal boron nitride substrate breaking the $\mC_2$ symmetry, the two flatbands split in energy, and we expect a chiral gapless edge state connecting them. The time-reversal partner of the chiral edge would occur from the other valley with an opposite chirality. Although these two flatbands are sublattice polarized, the chiral edge mode is sublattice unpolarized since it connects two bulk bands of opposite sublattice polarization.

\subsection{Zero-structure}
\label{nodesofzeromode}
We numerically observed that there are multiple zeros occurring at each order of magic angles, as shown in FIG.~\ref{plotwfnormthreeangles}. We now demonstrate they are indeed zeros rather than numerical artifacts.

We noticed that all extra zeros occurring at higher magic angles are located at the reflection symmetric lines. The mirror symmetry $\mM_y$ constrains that $\rho(x,y)$ and $\rho^*(x,-y)$ to be the same zero mode solutions. By using the global $U(1)$ phase degree of freedom of the wavefunction, $\rho(\bm r)$ can be chosen to be a purely real function on the reflection symmetric line $y$=$0$, the red dotted line of FIG.~\ref{illustrate_zeros}. Here we parameterize this line by $\bm r$=$\lambda(\bm a_1-\bm a_2)$=$(x,0)$ with $\lambda$$\in$$[-0.5,0.5)$, and plot $\rho(\lambda)$ along it in FIG.~\ref{plotaline}. Since $\rho(\lambda)$ is cell-periodic, it must cross zero along the reflection symmetric line an even number of times.

\begin{figure}[h]
    \centering
    \begin{subfigure}[b]{0.3\textwidth}
        \includegraphics[width=\textwidth]{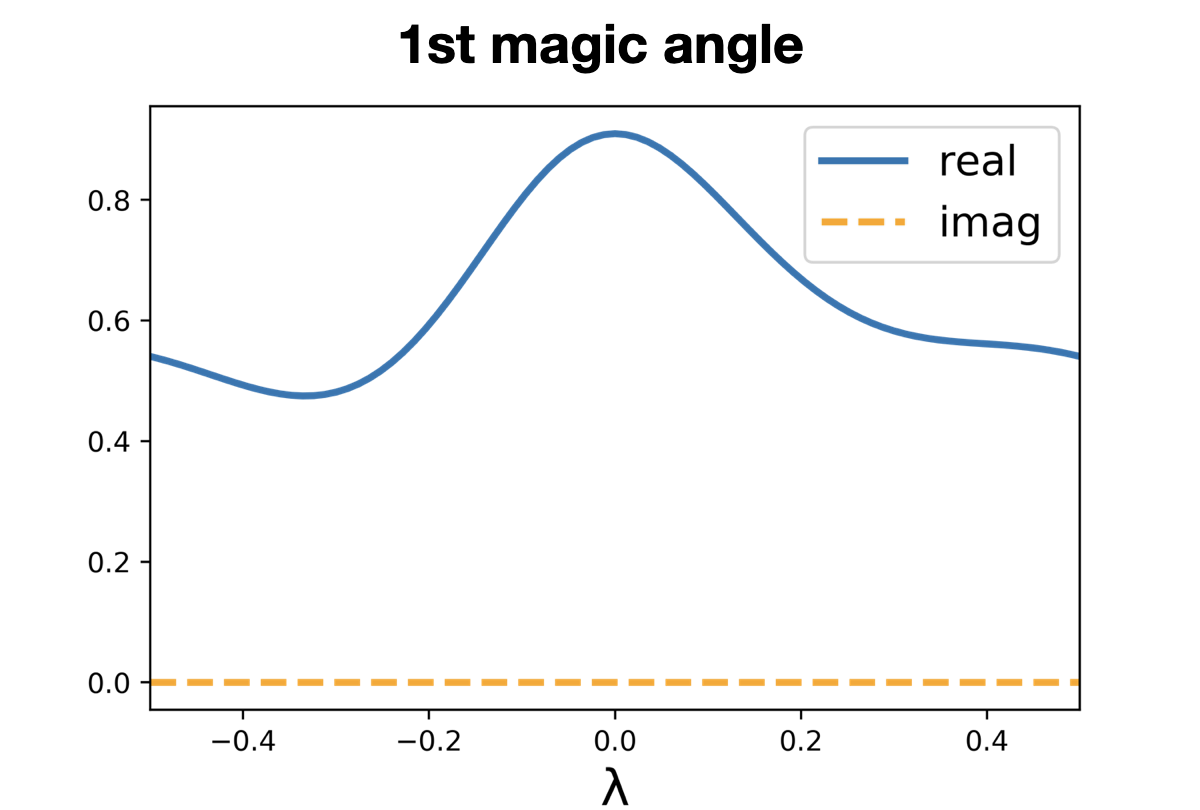}
    \end{subfigure}
    \begin{subfigure}[b]{0.3\textwidth}
        \includegraphics[width=\textwidth]{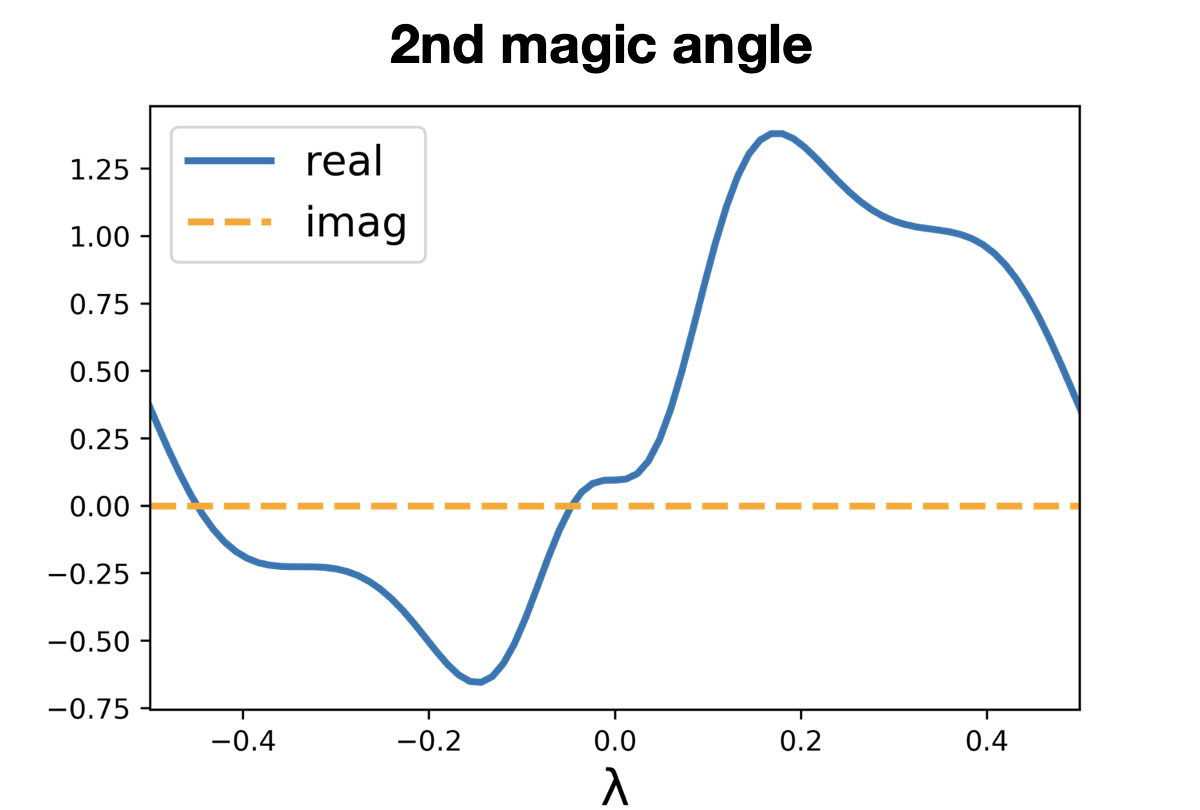}
    \end{subfigure}
    \begin{subfigure}[b]{0.3\textwidth}
        \includegraphics[width=\textwidth]{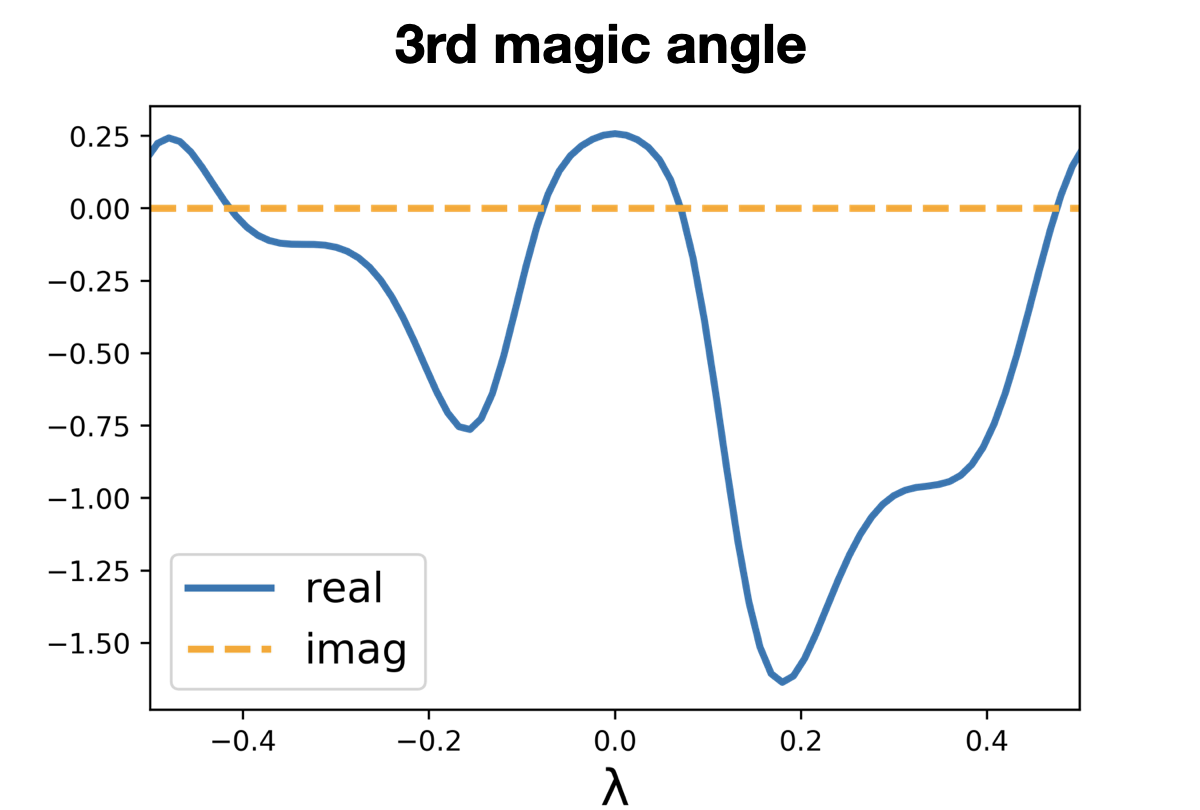}
    \end{subfigure}
    \caption{Plot of $\rho(\bm r)$ defined in Eqn.~(\ref{rhob}) along the reflection symmetric line parameterized by $(x,0)$=$\lambda(\bm a_1-\bm a_2)$ with $\lambda$$\in$$[-0.5,0.5)$. The blue solid and orange dashed lines indicate the real and imaginary part of $\rho$ respectively. It can be seen from these plots that $\rho(\lambda)$ is real, and crosses zero an even number of times.}\label{plotaline}
\end{figure}

The zeros of the zero mode wavefunctions are classified into two types by their ``movability''. One of them is a ``movable zero'' \cite{Arvas_movezero,KOHMOTO1985343,Brown_Nodal} from the quantum Hall wavefunction $\Phi_{\bm k}$, whose location moves linearly with the Bloch wavevector:
\begin{equation}
\bm r^a_{\bm k} = \bm r^a_0 + \epsilon^{ab}(\bm k-\bm K)_b.\label{loc_zk}
\end{equation}

This zero carries the external Hall response of the Chern band. There are other ``frozen zeros'' whose locations are fixed and independent of the \Blochk $\bm k$. In particular, among these frozen zeros, one of them is from the anti-quantum Hall state. In FIG.~\ref{illustrate_zeros}, we illustrate the zero-structure, where the black and blue dot represent the movable quantum Hall zero and the frozen anti-quantum Hall zeros. The yellow and red dots are frozen zeros from the function $\rho(\bm r)$.

Besides their ``movability'', zeros are also classified by their ``chirality'': the wavefunction receives a $2\pi n$ phase when the coordinate $\bm r$ encircles the zero once anticlockwise. We numerically noticed that the black and red dots are $n$=$1$ zeros, while yellow and blue are $n$=$-1$ zeros. Interestingly, as shown in FIG.~\ref{plotwfnormthreeangles} the center of the unit cells are concentrated with more and more $n$=$-1$ zeros at higher magic angles. We discuss the implication for circulating currents in the next section.

\begin{figure}[h]
    \centering
    \includegraphics[width=0.45\textwidth]{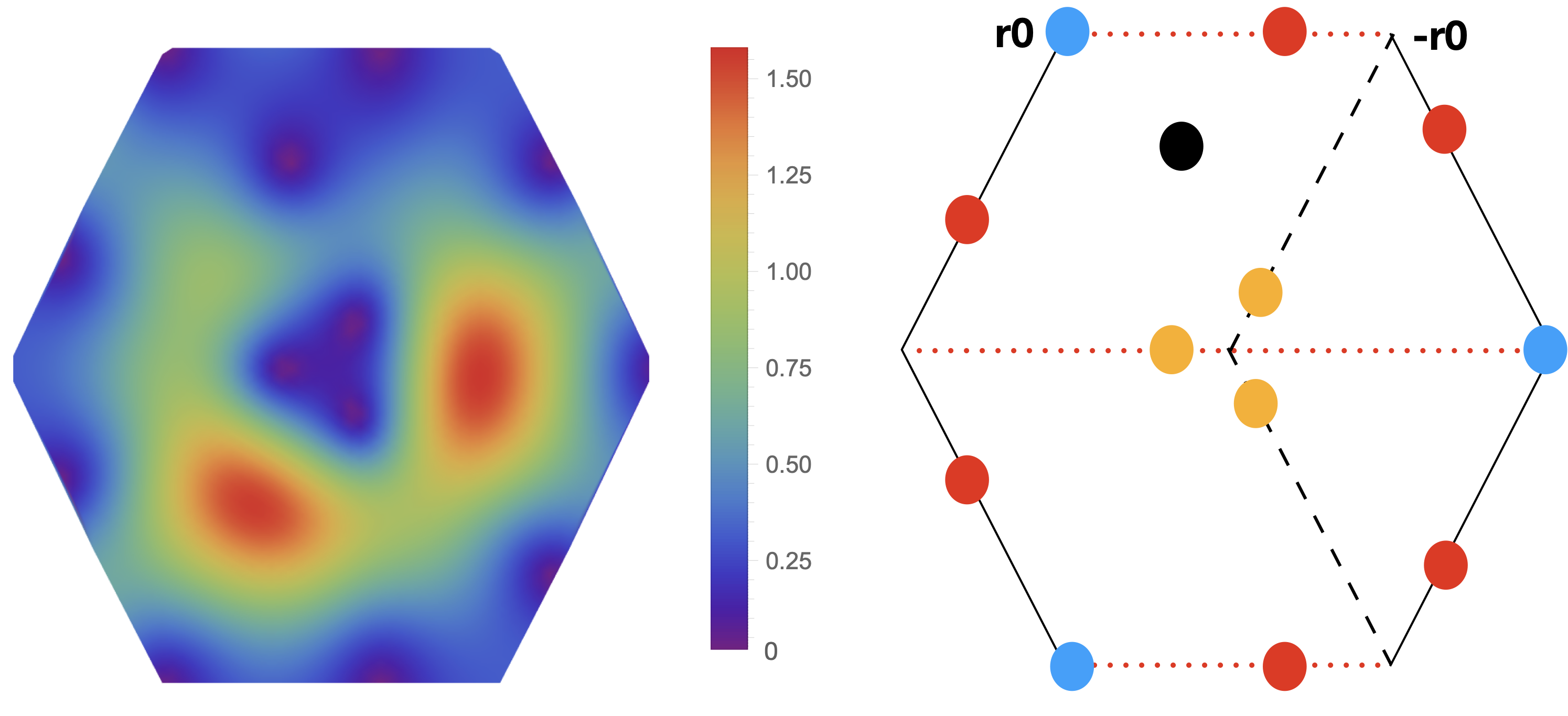}
    \caption{Illustration of zeros. Left: real space plot of the bottom component of the zero mode wavefunction $\phi^b_{\bm k}(\bm r)$ at the second magic angle. Right: sketch of the zeros of the same wavefunction. Without loss of generality, we here choose $\bm k$ to be a generic point, $\bm K$+$0.3\bm b_1$$-$$0.05\bm b_2$. Since location of the black dot (the zero of the quantum Hall wavefunction $\Phi_{\bm k}$) is locked with its Bloch momentum according to Eqn.~(\ref{loc_zk}), it called a movable zero. The blue dot is the frozen zero from the anti-quantum Hall wavefunction. The yellow and the red dots are the frozen zeros from the function $\rho(\bm r)$. We found the black and red dots are $n$=$1$ zeros, while the yellow and blue dots are $n$=$-1$ zeros, where $n$ is the $2\pi n$ phase that the wavefunction receives when the coordinate $\bm r$ encircles the zero once anticlockwise. The red dotted line is one of the three $\mC_3$ symmetry related reflection symmetric lines.}\label{illustrate_zeros}
\end{figure}

We have shown that the zero mode wavefunction shares some similarities with the simple harmonic oscillator system whose eigenstates also have alternating parity and have an increasing number of zeros. In Appendix~\ref{appE}, we provide an analytical argument why these features might persist for all higher magic angles by an analogy to the harmonic oscillator \cite{ODE_book}.

\section{Experimental Observation and Implications}
\label{Sec:ExperimentalObservation}
\subsection{Charge density and scanning tunneling probes}
One direct consequence of the zeros is a charge density deficiency that can be seen in scanning tunneling spectroscopy experiments \cite{TBG_exp_Andrei_chargeorder,Liang_CDW,Guinea_PNAS18}.

We expect a spectroscopy experiment will probe only the top (or bottom) layer, which corresponds to the components $\phi_t$ (top layer sublattice $A$ wavefunction) and $\chi_t$ (top layer sublattice $B$ wavefunction). If the spectroscopy measurement has spatial resolution on the level of the atomic spacing, then the sublattice wavefunctions can be probed separately. In this case, fixed zeros in the wavefunction components $\phi_t$ or $\chi_t$ correspond to the vanishing of charge density in real space, which will be strongly visible in the spectroscopy experiment. If the ground state is sublattice polarized, which maybe the case on a hexagonal boron nitride substrate, then such spatial resolution is not required to observe the zeros in a spectroscopy experiment. Notice that since the opposite valley wavefunction on the same sublattice is related by $\mT$, which acts trivially in real space, we expect the two opposite valley wavefunctions on the same sublattice have the same location of fixed zeros. Therefore, probing the zeros with spectroscopy does not require valley polarization.

The accumulation of zeros at the unit cell center and the unit cell boundary as magic angle order increases, as shown in FIG.~\ref{plotwfnormthreeangles}, should also be visible by spectroscopy with even less atomic resolution. This accumulation will become more prominent at higher magic angles.

Away from the chiral model, the zeros become non-zero minima in the charge density, which we have observed numerically. These will give a less sharp signature in scanning tunneling spectroscopy experiments, but will likely still be observable over some parameter regime.

\subsection{Higher Landau level physics at higher magic angles}
As we have seen from Section~\ref{sec:structure_zeros}, the $\mG(\bm r)$ piece of the flatband wavefunction in Eqn.~(\ref{zeromodewavefunction}) has an increasing number of zeros and has an analytical expansion similar to an anti-quantum Hall wavefunction. Consequently, we interpreted the zero mode wavefunction as a product of a higher Landau level anti-quantum Hall state and a lowest Landau level quantum Hall state (Eqn.~(\ref{rhob})), where the Landau level index of the former is determined by the order of the magic angle. We also discussed in Section~\ref{revisit_zeromode_eqn} that the topological properties of the flatbands are determined by the lowest Landau level quantum Hall piece $\Phi_{\bm k}$ since $\mG(\bm r)$ does not have \Blochk dependence.

Nevertheless, we expect the effective interactions projected into the flatbands are modified strongly by both $\mG(\bm r)$ and $\Phi_{\bm k}$. In particular, the nodal structure of $\mG(\bm r)$ directly impacts the charge density, which determines the projected Coulomb interaction. In the quantum Hall problem, the nodal structure of the higher Landau levels results in arrangement of charge that ultimately stabilizes various states \cite{Fogler_CDW,RezayiHaldaneCFL,KunYang_WignerCrystal,PhysRevB.93.201303} such as charge density waves, bubble phases, and other many-body topological phases (for instance, the non-Abelian Moore-Read phase \cite{MoreReadState}) that are absent in the lowest Landau level. By analogy, we might expect a different set of interacting phases to be stabilized at higher magic angles than at the first magic angle. Our formulation provides a theoretical and computational pathway towards analyzing interacting physics at different magic angles.

\subsection{Local current and magnetization at higher magic angles}
From the charge density of the zero mode wavefunction as plotted in FIG.~\ref{plotwfnormthreeangles}, we observe that for the first magic angle, the charge density maximum occurs at the unit cell center, {\it i.e.} the \emph{AA} stacking point. At higher magic angles, we see an increasing number of zeros appearing at this region. Interestingly, all these zeros are of the same chirality for both layers, while zeros of the opposite chirality are pushed to the boundary of the unit cell. This indicates a stronger phase winding effect and hence circulating currents near the \emph{AA} stacking region at higher magic angles, which could be experimentally observable nearby the chiral limit.

To see the circulation currents, we first define the following intra-sublattice intra-layer ``current operator'' $\bm J^l_{ss}$ for sublattice $s$ and layer $l$. The operator $\bm J^l_{AA}$ is defined as:
\begin{equation}
\bm J^l_{AA}(\bm r) \equiv i(t')[(\bm\nabla\phi_{l})^*\phi_{l}-\phi_{l}^{*}(\bm\nabla\phi_{l})](\bm r).\label{JAA_operator}
\end{equation}
operator $\bm J^l_{BB}(\bm r)$ is defined in a similar manner but with $\phi$ replaced by $\chi$. Here $t'$ is the microscopic parameter representing the next-nearest-neighbor hopping strength. We call the above a ``current operator'' in quot because $\bm J^l_{ss}$ is not the current operator of the chiral model, which by definition should be proportional to $\partial_{\bm k}H_{\bm k}$, and hence couples distinct sublattices and vanishes within one sublattice. Since the exact flatband wavefunctions are fully sublattice polarized (corresponding physically to hexagonal boron nitride splitting the sublattice degeneracy), the current operator $\partial_{\bm k}H_{\bm K}$ vanishes within one sublattice polarized state. Nevertheless, we argue that our current operator $\bm J^l_{ss}$ has a microscopic origin, and hence should be a physical current operator. The $\bm J^l_{ss}$ can be regarded as a continuum version of lattice current $i(a_{s,i}^{\dag}a_{s,j}-a_{s,j}^{\dag}a_{s,i})$ induced from the next-nearest-neighbor hopping process in graphene, where the $a_{s,i}$ are graphene's electron annihilation operators and $i,j$ labels graphene's next-nearest-neighbor sites.
\begin{figure}[h]
    \centering
    \includegraphics[width=0.5\textwidth]{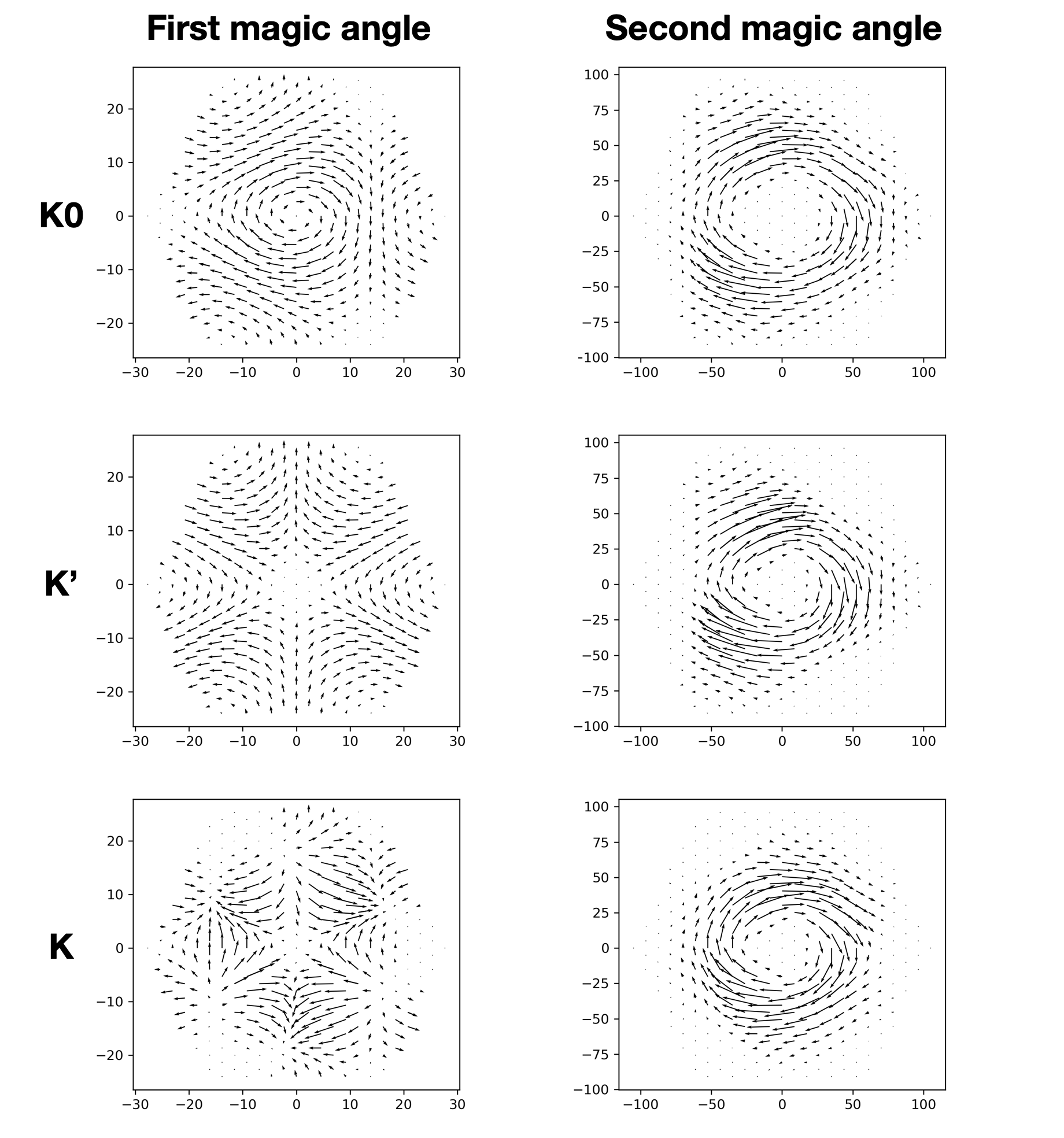}
    \caption{Expectation value of the bottom layer sublattice $A$  current operator $\bm J^b_{AA}$ at the \mr Gamma point $\bm K_0$ and \mr Dirac points $\bm K'$, $\bm K$ at the first two magic angles. The chirality of the current operator would be opposite on the B sublattice, and the chirality is the same for both layers within one sublattice.}\label{loopcurrent}
\end{figure}

In FIG.~\ref{loopcurrent}, we plot the real space distribution of $\bm J^b_{AA}(\bm r)$, calculated from the bottom layer sublattice-$A$ wavefunction $\phi_b$ at three different Bloch momenta. According to Ref.~(\onlinecite{Law_TBG_OM}), the orbital magnetization is dominated by the Gamma point $\bm K_0$ and Dirac point $\bm K'$ in a single valley model, since the bands hybridize most strongly with other bands at these points. Given the strong circulating current present at the second magic angle, it is reasonable to speculate a stronger orbital magnetization \cite{Vanderbilt_OM_periodic,Vanderbilt_OM_multiband,Vanderbilt_StructuralResponse,Di_BerryDOS,Niu_QuantumTheoryOM,Di_Review,om_review} at higher magic angles than the magnetization at the first magic angle \cite{XiDai_OM,Law_TBG_OM,experiment_om_tbg} for cases close to the chiral limit. Note that the circulation currents are odd under time reversal or a sublattice transformation, hence a strong experimental signal requires valley and sublattice polarization. We leave a detailed exploration of higher magic angle orbital magnetization with more realistic parameters as future work.

\section{Conclusion}\label{sec:conclusion}
In this work, we studied the chiral model of twisted bilayer graphene introduced in Ref.~(\onlinecite{Grisha_TBG}). We pointed out the intrinsic \ourinv symmetry of the chiral model, protected by the $\mC_2\mT$ crystal symmetry and the linearized Dirac fermion. As a consequence, the energy spectrum is inversion symmetric at all twist angles. Furthermore, zero modes occurring at different magic angles are distinguished by their \ourinv eigenvalue. We numerically found a correspondence of the zero mode inversion parity and the order of magic angles and speculated such an alternating pattern would hold for all magic angles.

We also pointed out the intricate relation between the zero mode wavefunction and the quantum Hall wavefunctions. As guaranteed by \ourinv symmetry, the zero mode wavefunction has an internal spinor structure, and in fact each component can be regarded as a product of a quantum Hall and an anti-quantum Hall wavefunction, which guarantees the zero mode has the periodicity of a Bloch wavefunction. Interestingly, there are an increasing number of zeros occurring in each component at higher magic angles.

In the end, we discussed the implications of our results to realistic systems and observable phenomena. First, these zeros can be detected as charge minima in real space by scanning tunneling spectroscopy. Second, the increasing number of zeros present in the zero mode wavefunction resembles the increasing number of zeros present in higher Landau level wavefunctions. Motivated by this observation, we anticipate higher Landau level physics will occur at the second and higher magic angles. Moreover, we noticed the phase circulation of the flatband wavefunctions at higher magic angles, and anticipate phenomena related to magnetization. We leave more detailed studies on higher magic angles as future work. We notice an earlier work on low twist angle physics in Ref.~(\onlinecite{zzy_naturecom20}).

Last but not least, it is well known that on a compact manifold, a $U(1)$ magnetic field is subject to a Dirac quantization condition \cite{Aharonov_Casher79}. Our identification of zero mode wavefunctions with two quantum Hall wavefunctions may also shed light on the non-Abelian quantization condition \cite{Non_Abelian_Graphene,chiral_vortexlattice,Kailasvuori_2009}, where a semi-classical analysis was done recently in Ref.~(\onlinecite{RafeiRen_TBG}).

\emph{Note added}: during the final stage of the manuscript, we noticed the ``unitary particle-hole'' symmetry occurring in Ref.~(\onlinecite{song2020tbg}) which is similar but distinct from our \ourinv symmetry in Eqn.~(\ref{involution}); we contrast the difference between the two in the paragraphs under Eqn.~(\ref{Hpseudoinv}). We also noticed a relevant work on the chiral model, Ref.~(\onlinecite{popov2020hidden}), that appeared recently.

\section*{Acknowledgments}
The Flatiron Institute is a division of the Simons Foundation. This work was partially supported by the Air Force Office of Scientific Research under Grant No.~FA9550-20-1-0260 (J.C.).

\section*{Appendices}
\appendix
\section{Model Hamiltonians and Unitary Transformations}\label{Utransformations}
\subsection{Lattices}
We start with setting up the notation of \mr lattice. As mentioned in the main text, we denote the two dimensional lattice vectors as $\bm a_{i=1,2}^{a=x,y}$. The area of unit cell is defined to be $2\pi S$:
\begin{equation}
2\pi S \equiv |\bm{a}_1\times\bm{a}_2| = |\epsilon_{ab}\bm{a}_1^a\bm{a}_2^b|.\label{area}
\end{equation}
where $\epsilon_{xy}$=$-\epsilon_{yx}$=1 is the anti-symmetric symbol. The reciprocal basis vectors are:
\beqn
\bm{b}^i_a = \epsilon^{ij}\epsilon_{ab}\bm{a}_j^b/S.\label{def_b}
\eeqn

As will be shown in Appendix~\ref{sec:QH_wavefunction}, $\sqrt S$ defines an effective magnetic length. We set $\sqrt S$=$1$ throughout this work.

Graphene contains $A$ and $B$ sites. As shown in FIG.~\ref{mBZ} of the main text, the Dirac points $\bm K/\bm K'$ and $A/B$ sites are located at,
\beqn
\bm K &=& \frac{-2\bm{b}_1+\bm{b}_2}{3},\quad \bm K' = \frac{2\bm{b}_2-\bm{b}_1}{3}.\nonumber\\
\bm r_A &=& \frac{\bm a_1+2\bm a_2}{3},\quad \bm r_B = \frac{2\bm a_1+\bm a_2}{3}.\label{def_r0_K}
\eeqn
We use $\bm r_0$ for $\bm r_A$ throughout this work.

\subsection{Unitary transformations}
In Section~\ref{sec:TBGreview}, we described the Bistritzer-MacDonald Hamiltonian Eqn.~(\ref{HTBG}) and the chiral model Eqn.~(\ref{chiral-form}) of a single valley. They are written in the non-chiral $\Psi_{BM}$ and the chiral basis $\Psi_c$ respectively, see Eqn.~(\ref{macdonaldbasis}) and Eqn.~(\ref{chiralbasis}). In this section, following Ref.~(\onlinecite{Grisha_TBG}), we work out the details of the unitary transformation between the two bases. We start with the continuum model Eqn.~(\ref{HTBG}), and perform a gauge transformation to remove the momentum shift on the diagonal. The Hamiltonian is transformed to be:
\begin{eqnarray}
H_{BM} &=& M_T\left(\begin{matrix}-iv_0\bm{\sigma}_{+\theta/2}\cdot\bm{\nabla}&T(\bm r)\\T^{\dag}(\bm r)&-iv_0\bm{\sigma}_{-\theta/2}\cdot\bm{\nabla}\end{matrix}\right)M_T^{\dag},\nonumber\\
M_T &=& \diag(e^{i\bm K_+^b\cdot\bm r},~e^{i\bm K_+^t\cdot\bm r}).
\end{eqnarray}
where $T(\bm r)$ is given in Eqn.~(\ref{T}). Then, we remove the diagonal $\theta$ dependence by rotation:
\begin{eqnarray}
H_{BM} &=& (M_TM_{\theta})H_{cBM}(M_TM_{\theta})^{\dag}.
\end{eqnarray}
where
\begin{eqnarray}
H_{cBM} &=& \left(\begin{matrix}-iv_0\bm{\sigma}\cdot\bm{\nabla}&T(\bm r)\\T^{\dag}(\bm r)&-iv_0\bm{\sigma}\cdot\bm{\nabla}\end{matrix}\right),\nonumber\\
M_{\theta} &=& \diag(e^{-\frac{i\theta}{4}\sigma_z},~e^{\frac{i\theta}{4}\sigma_z}).
\end{eqnarray}

The matrix $H_{cBM}$ is the chiral Hamiltonian organized in basis $(\phi_b,\chi_b,\phi_t,\chi_t)^T$ where $\phi$ and $\chi$ represent the $A$ and $B$ sublattice respectively, and $b/t$ represent the bottom and top layer components. More explicitly,
\beqn
H_{cBM} = \sqrt2v_0\left(\begin{matrix}0&-i\partial&0&\alpha U_{-\phi}\\-i\bar{\partial}&0&\alpha U_{\phi}&0\\0&\alpha U^*_{\phi}&0&-i\partial\\\alpha U^*_{-\phi}&0&-i\bar{\partial}&0\end{matrix}\right).
\eeqn
where $U_{\phi}$ is defined in Eqn.~(\ref{Uphi}). Transforming into the chiral basis Eqn.~(\ref{chiralbasis}), we obtain Eqn.~(\ref{chiral-form}):
\begin{eqnarray}
H_{cBM} &=& v_0\left(\begin{matrix}0&\mathcal{D}\\\mathcal{D}^{\dag}&0\end{matrix}\right),\\
\mathcal{D}^{\dag} &=& \sqrt2\left(\begin{matrix}-i\bar{\partial}&\alpha U_{\phi}\\\alpha U^*_{-\phi}&-i\bar{\partial}\end{matrix}\right),\quad\mathcal{D} = \sqrt2\left(\begin{matrix}-i\partial&\alpha U_{-\phi}\\\alpha U^*_{\phi}&-i\partial\end{matrix}\right).\nonumber
\end{eqnarray}

The unitary transformation Eqn.~(\ref{unitarytrans}) can be easily worked out from matrices $M_T$, $M_{\theta}$ and the basis shuffling. As defined in the main text, we denote the rotated graphene Dirac points as $\bm K_+^{b/t}$, and denote the \mr Dirac points as $\bm K$=$\bm K_+^b$$-$$\bm K_+^\Gamma$, $\bm K'$=$\bm K_+^t$$-$$\bm K_+^\Gamma$ where $\bm K_+^\Gamma$ is the \mr Brillouin zone center. We have also shifted the \Blochk of the chiral basis Eqn.~(\ref{chiralbasis}) to center at the \mr Gamma point. Its Bloch translation symmetry can be also worked out easily as shown in Eqn.~(\ref{bc_chiralbasis}).

\section{How $\mC_2\mT$ Symmetry Constrains the Chiral Hamiltonian}
\label{sec:append_sym}
We have written the inter-layer coupling matrix in real space as Eqn.~(\ref{HTBG}). We now discuss the action of $\mC_2\mT$, which complex conjugates and exchanges the two sublattices. The diagonal blocks in Eqn.~(\ref{HTBG}) are invariant under this transformation. We now consider the off-diagonal tunneling terms $H_{BM}^{\rm tun}$. Its transformation under $\mC_2\mT$ reads,
\beqn
H_{BM}^{\rm tun} &\xrightarrow{\mC_2\mT}& \int d\bm{r} \Psi^\dagger(-\bm{r})  \begin{pmatrix} 0 & \sigma_x T(\bm{r})\sigma_x \\ \sigma_x T^\dagger(\bm{r})\sigma_x & 0 \end{pmatrix}^*  \Psi(-\bm{r}) \nonumber\\
&=&  \int d\bm{r} \Psi^\dagger(\bm{r})  \begin{pmatrix} 0 & \sigma_x T(-\bm{r})\sigma_x \\ \sigma_x T^\dagger(-\bm{r})\sigma_x & 0 \end{pmatrix}^*  \Psi(\bm{r}),\nonumber
\eeqn
where, same as the main text, the Pauli matrices $\bm\sigma$ act on sublattice space. By virtue of being invariant under $\mC_2\mT$, it follows that:
\begin{equation}
T(\bm{r}) = \sigma_x T^*(-\bm{r})\sigma_x,
\end{equation}
or, element by element:
\begin{equation}
T_{AA}(\bm{r}) = T_{BB}^*(-\bm{r}), \,\, T_{AB}(\bm{r}) = T_{BA}^*(-\bm{r}).
\label{eq:C2TonTels}
\end{equation}

If we rotate to the chiral basis $\Psi_c$ Eqn.~(\ref{chiralbasis}), the tunneling terms enter in the following way:
\begin{equation}
\mathcal{H}_{cBM}(\bm{r}) = \begin{pmatrix} 
T^{\rm diag}_{A}(\bm{r}) & \mathcal{D}(\bm{r}) \\ \mathcal{D}^\dagger(\bm{r}) & T^{\rm diag}_{B}(\bm{r})
\end{pmatrix},
\end{equation}
where the diagonal blocks are:
\beqn
T^{\rm diag}_{A}(\bm r) &=& \left(\begin{matrix}&T_{AA}(\bm r)\\T_{AA}^*(\bm r)&\end{matrix}\right),\nonumber\\
T^{\rm diag}_{B}(\bm r) &=& \left(\begin{matrix}&T_{BB}(\bm r)\\T_{BB}^*(\bm r)&\end{matrix}\right).
\eeqn
and the off-diagonal blocks are:
\beqn
\mathcal{D}(\bm{r}) &=& \begin{pmatrix} -\sqrt2i\partial & T_{BA}^*(\bm{r}) \\ T_{AB}(\bm{r}) & -\sqrt2i\partial \end{pmatrix},\nonumber\\
\mathcal{D}^\dagger(\bm{r}) &=& \begin{pmatrix} -\sqrt2i\bar{\partial} & T_{AB}^*(\bm{r}) \\ T_{BA}(\bm{r}) & -\sqrt2i\bar{\partial} \end{pmatrix}.
\eeqn

Using the action of $\mC_2\mT$ in Eqn.~(\ref{eq:C2TonTels}), these can be written in terms of only one complex parameter $T_{AB}$:
\beqn
\mathcal{D}(\bm{r}) &=& \begin{pmatrix} -i\sqrt2\partial & T_{AB}(-\bm{r}) \\ T_{AB}(\bm{r}) & -\sqrt2i\partial \end{pmatrix},\nonumber\\
\mathcal{D}^\dagger(\bm{r}) &=& \begin{pmatrix} -\sqrt2i\bar{\partial} & T_{AB}^*(\bm{r}) \\ T_{AB}^*(-\bm{r}) & -\sqrt2i\bar{\partial}\end{pmatrix},\label{eq:C2TonD}
\eeqn
which satisfy our Lemma~\ref{lemmaDsigmay}:
\begin{eqnarray}
\tau_y\mD^{\dag}(\bm r)\tau_y &=& -\mD^{\dag}(-\bm r),\nonumber\\
\tau_y \mD(\bm{r})\tau_y &=& -\mD(-\bm{r}),\label{eq:sigmayrelation}
\end{eqnarray}
where the Pauli matrices $\tau$ act on the layer index.

In this basis, the chiral matrix $\sigma_z$ that anti-commutes with the Hamiltonian enforces $T_{AA}(\bm r)$=$T_{BB}(\bm r)$=$0$. Notice that Eqn.~(\ref{eq:sigmayrelation}) also requires linearized Dirac fermion; a quadratic term in the dispersion destroys it. Note that a quadratic term in the dispersion also destroys the exact flatband of the chiral model. We hence demonstrated that for chiral models with linearized Dirac fermion, \ourinv follows from $\mC_2\mT$ symmetry.

\section{Quantum Hall Wavefunction}\label{sec:QH_wavefunction}
In this section we review the quantum Hall wavefunction that is frequently used in the main text. We start with discussing magnetic translation symmetry and quasi-periodic elliptic functions.

\subsection{Magnetic translation symmetry}
Since the lowest Landau level wavefunctions are usually written in terms of holomorphic functions, we start by setting up a notation for complex coordinates. Complex structures $\omega_{a=x,y}$ and $\omega_{a=x,y}^*$ define a one-to-one mapping from two dimensional affine space to the complex plane. We represent the metric and the anti-symmetric tensor as
\beqn
g_{ab} &=& \omega^*_a\omega_b+\omega_a\omega_b^*,\nonumber\\
i\epsilon_{ab} &=& \omega^*_a\omega_b-\omega_a\omega_b^*.\label{def_g_omega}
\eeqn

They have the properties: $\omega^a=g^{ab}\omega_b$, $\omega^a\omega_a=0$, and $\omega^a\omega_a^*=1$. The complex vectors are defined by contracting complex structure with vectors $A$$\equiv$$\omega_a\bm A^a$, and complex co-vectors as $B$$\equiv$$\omega^a\bm B_a$. To distinguish with vectors, complex vectors are unbold. In terms of complex coordinates, the inner product and cross product are respectively $\bm A$$\cdot$$\bm B$$\equiv$$\bm A_a\bm B^a$=$AB^*$+$A^*B$, $\bm A$$\times$$\bm B$$\equiv$$\epsilon_{ab}\bm A^a\bm B^b$=$-i(A^*B-AB^*)$. In this work, we took $\omega_x=1/\sqrt2$ and $\omega_y=i/\sqrt2$.

The quantum Hall system describes two dimensional interacting or noninteracting electrons in a perpendicular magnetic field. In a magnetic field, the electron's coordinate is factorized into the center of its cyclotron motion {\it i.e.} guiding center $\bm R$, and the radius {\it i.e.} Landau orbits $\bm{\bar R}$:
\begin{equation}
\bm r = \bm R + \bm{\bar R},
\end{equation}
where $\bm R$ commutes with $\bm{\bar R}$, but their individual components are noncommutative:
\beqn
~[\bm R^a,\bm R^b] = -i\epsilon^{ab}l_B^{2},\quad[\bar{\bm R}^a,\bar{\bm R}^b] = i\epsilon^{ab}l_B^{2}.
\eeqn

In our case the area of unit cell $S$ plays the same role as magnetic length squared $l_B^2$=$\hbar/|eB|$ where $e,B$ are electron charge and magnetic field strength. When projected into a single Landau level, an electron is fully described by the noncommutative $\bm R$ degrees of freedom. The magnetic translation operator is defined as the following one,
\beqn
t(\bm d) \equiv \exp(i\bm d\times\bm R),
\eeqn
which translates the guiding center $\bm R$ by distance $\bm d$. The magnetic translation algebra is,
\beqn
t(\bm d_1)t(\bm d_2) &=& t(\bm d_2)t(\bm d_1)e^{i\bm d_1\times\bm d_2} = t(\bm d_1+\bm d_2)e^{\frac{i}{2}\bm d_1\times\bm d_2}.\nonumber
\eeqn

Due to the single value of wavefunction, any legal wavefunction must transform back to itself after a periodic translation. So we have the boundary condition,
\begin{equation}
t(\bm a)\psi = e^{i\phi_{\bm a}}\psi,\label{quantumhallbc}
\end{equation}
where $\bm a\in\mathbb{A}$ is a lattice vector. From now on we define the whole lattice as $\mathbb{A}$$\equiv$$\{m\bm a_1+n\bm a_2|m,n\in\mathbb{Z}\}$. The phase factor $\phi_{\bm a}$ effectively measures the fraction of flux inside the torus. The wavefunctions that satisfy Eqn.~(\ref{quantumhallbc}) are written in terms of elliptic functions. One choice of elliptic function is the Jacobi theta function \cite{haldanetorus1}. Recently it was also found that the ``modified Weierstrass sigma function'' is another choice \cite{haldanemodularinv,Jie_MonteCarlo}. Compared with Jacobi theta function, Weierstrass sigma function has the advantage of being modular invariant.

\subsection{Modified Weierstrass sigma function}
The modified Weierstrass sigma function \cite{haldanemodularinv,Jie_MonteCarlo} $\sigma(z)$ is defined as:
\beqn
\sigma(z) &=& \tilde\sigma(z)e^{-\frac{1}{2}\bar G(\mathbb A) z^2},\label{defofsigma}
\eeqn
{\it i.e.} a product of the standard Weierstrass sigma function $\tilde\sigma(z)$ and a holomorphic factor $e^{-\frac{1}{2}\bar G(\mathbb A) z^2}$, where as will be explained soon the ``almost modular form'' $\bar G(\mathbb A)$ is a modular independent $c-$number constant that vanishes for square and hexagonal torus. We now introduce the $\bar G(\mathbb{A})$, and discuss the quasi-periodic property of $\sigma(z)$.

The standard Weierstrass sigma function $\tilde\sigma(z)$ has a product series expansion (which is also a fast converging form for numerics),
\beqn
\tilde\sigma(z) \equiv z\prod_{a\in A_{mn}\backslash\{0\}}\left(1-\frac{z}{a}\right)e^{\frac{z}{a}+\frac{1}{2}\frac{z^{2}}{a^{2}}},
\eeqn
where as defined above, $\mathbb{A}$ means the set of lattice points. Clearly, it is modular invariant. It is also quasi-periodic,
\beqn
\tilde\sigma(z+a_i) = - e^{2\tilde\eta_i(z+a_i/2)}\tilde\sigma(z),\quad i=1,2,\nonumber
\eeqn
where $\tilde\eta_i$ is the standard zeta function evaluated at half period, which is related to the $k$=1 Eisenstein series $G_2(a_i)$, $i=1,2$,
\beqn
\tilde\eta_i &=& G_2(a_i)a_i/2.\label{eta-def}
\eeqn

The Eisenstein series $G_2(a_i)$ has a highly convergent formula,
\beqn
G_2(a_i) &=& \frac{2\pi^2}{a_i^2}\left(\frac 1 6 + \sum_{n=1}^{\infty}\frac{1}{\sin^2(n\pi\frac{a_{j\neq i}}{a_i})}\right).
\eeqn

The $\tilde\eta_i$ in addition obey a relation that defines chirality,
\beqn
\tilde\eta_1a_2 - \tilde\eta_2a_1 = \frac{1}{2N_{\phi}}(a_1^*a_2 - a_1a_2^*) = i \pi. \label{chirality}
\eeqn

In our case, the magnetic flux quanta of a unit cell is one, so $N_{\phi}$=1. The (\ref{eta-def}) and (\ref{chirality}) suggests a modular independent quantity called ``almost modular form'',
\beqn
\bar G(\mathbb A) \equiv G_2(a_i) - \frac{1}{N_{\phi}}\frac{a_i^*}{a_i}.
\eeqn

With these formulas in hand, we are ready to get the quasi-periodicity of $\sigma(z)$:
\begin{equation}
\sigma(z+a_i) = -e^{a_i^*(z+a_i/2)}\sigma(z),\quad i=1,2.
\end{equation}

Last but not least, the sigma function is odd under spatial inversion: $\sigma(-z)=-\sigma(z)$.

\subsection{Quantum Hall wavefunction}\label{sec:qhwf}
The quantum Hall wavefunction is given in Eqn.~(\ref{quantumhallwf}), which we copy below:
\begin{equation}
\Phi_{\bm k}(\bm r) = e^{z_k^*z}\sigma(z-z_k)e^{-\frac12|z_k|^2}e^{-\frac12|z|^2}.\nonumber
\end{equation}

It has a single zero located at $\bm r^a_{\bm k}$=$\bm r^a_0 + \epsilon^{ab}(\bm k-\bm K)_b$ in each unit cell, with $\bm r_0$ defined in Eqn.~(\ref{def_r0}). Mapping to the complex plane, the zero occurs at $z_k$ and its translated counterparts, where $z_k$ is:
\beqn
z_k &=& \omega_a(\bm r_0^a-\epsilon^{ab}\bm K_b) + \omega_a\epsilon^{ab}\bm k_b = -ik,\label{comp_zk}
\eeqn
where the first term is zero following from Eqn.~(\ref{def_b}) and Eqn.~(\ref{def_r0_K}). We used Eqn.~(\ref{def_g_omega}) to derive the second term.

Since $\Phi_{\bm k}$ is not a Bloch function, the ``Bloch vector'' $\bm k$ here should be understood as labeling the magnetic translation boundary condition Eqn.~(\ref{quantumhallbc}): $t(\bm a_{1,2})\Phi_{\bm k}$=$-e^{i\bm k\cdot\bm a_{1,2}}\Phi_{\bm k}$. For a quantum Hall wavefunction, its zero moves linearly with the boundary condition $\bm k$, reflecting the fact of Chern number $\mC$=$1$ \cite{Arvas_movezero}. The following diagram FIG.~\ref{mapingrk} is helpful to quickly figure out $\bm r_{\bm k}$ given the \Blochk $\bm k$.
\begin{figure}[h]
    \centering
    \includegraphics[width=0.3\textwidth]{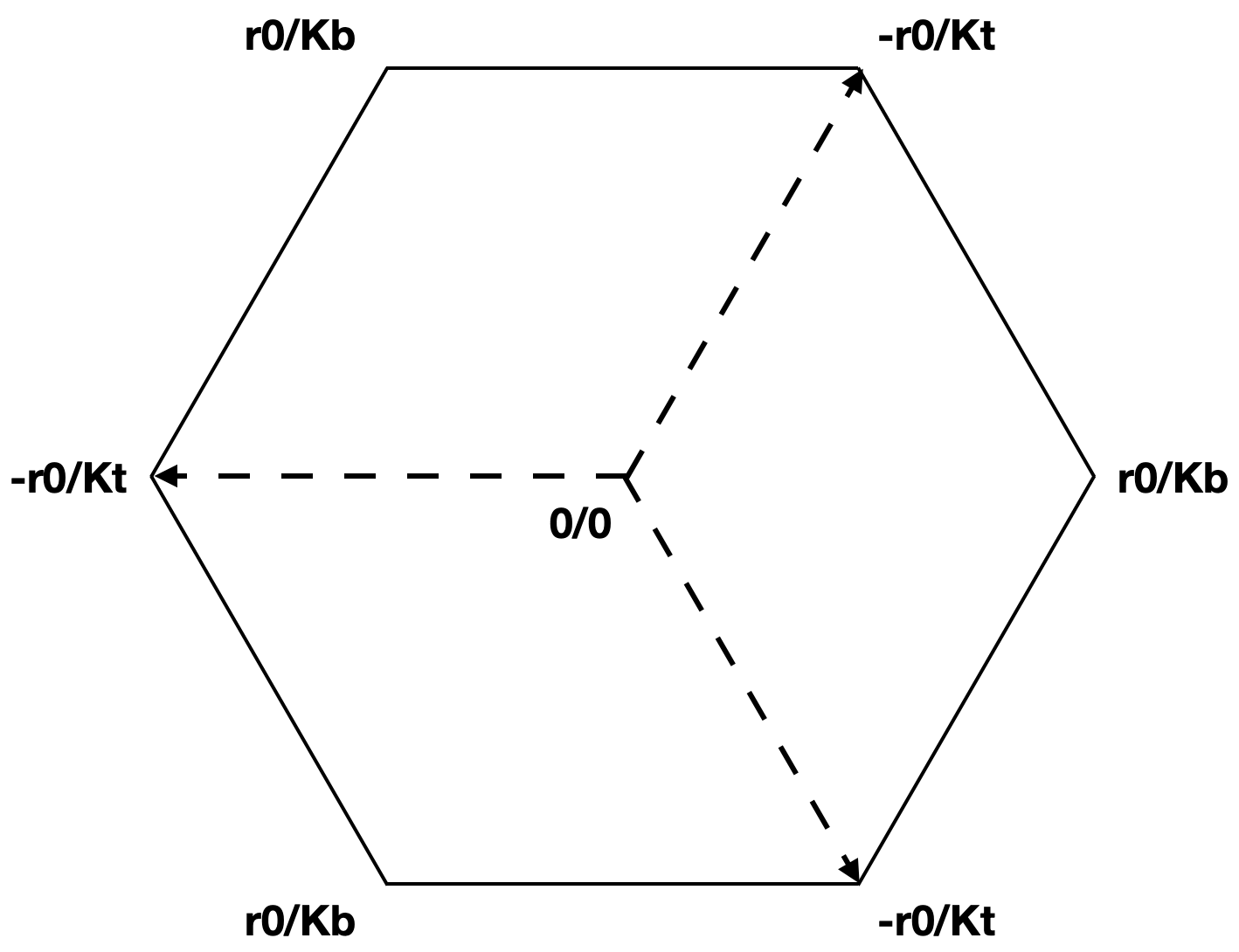}
    \caption{The one-to-one mapping between $\bm k$ and $\bm r_{\bm k}$, where the first and second letter are $\bm r_{\bm k}$ and $\bm k$ respectively. The figure is constructed by rotating the \mr Brillouin zone by 90 degrees and overlaps with the real space unit cell, precisely because of the mathematical relation $\bm r_{\bm k}^a$=$\bm r_0^a$+$\epsilon^{ab}(\bm k-\bm K)_b$. All points in the diagram are illustrated modulo lattice vectors.}\label{mapingrk}
\end{figure}

Using results derived in the last section, it is easy to find the quasi-periodic boundary condition $\Phi_{\bm k}$ satisfies in real and reciprocal space. With $i$=1,2, they are:
\beqn
\Phi_{\bm k}(\bm r+\bm a_i) &=& -e^{\frac{i}{2}\bm a_i\times\bm r}e^{i\bm r_{\bm k}\times\bm a_i}\Phi_{\bm k}(\bm r).\label{QHbc}\\
\Phi_{\bm k+\bm b_i}(\bm r) &=& -e^{\frac{i}{2}\bm b_i\cdot\bm r_{\bm k}}\Phi_{\bm k}(\bm r).\label{QHkbc}
\eeqn

The nontrivial phase factors above cannot be removed by a smooth, global, gauge transformation, which reflects the fact that $\Phi_{\bm k}$ has a nontrivial Chern number. Technically, these boundary conditions allow one to restrict the discussion to the unit cell and the first Brillouin zone. From now on, we denote $\bm k$ as \Blochk inside the first Brillouin zone.

It is straightforward to see how inversion acts on quantum Hall wavefunctions from Eqn.~(\ref{comp_zk}):
\begin{equation}
\Phi_{\bm k}(\bm r) = -\Phi_{-\bm k}(-\bm r).\label{QHinv}
\end{equation}

We finish this section by showing the boundary condition of $\mG(\bm r)$, which follows straightforwardly from the periodicity of the zero mode wavefunction in Eqn.~(\ref{bc_chiralbasis}) and the quantum Hall wavefunction in Eqn.~(\ref{QHbc}):
\begin{equation}
\mathcal{G}(\bm r+\bm a_{i=1,2}) = -\mathcal{G}(\bm r)\times e^{-\frac{i}{2}\bm a_i\times\bm r}e^{i\bm q_0\cdot\bm a_i}.\label{boundary-condition-G}
\end{equation}

\section{Analytical Argument for the Nodal Structure}\label{appE}
The alternating parities and the increasing number of zeros we observed in the chiral model shares many similarities as the simple harmonic oscillator. In this section, we provide an argument for the zero-structure and inversion patterns by making an analogy to simple one-dimensional harmonic oscillators. Specifically, since the additional zeros that occur at higher magic angles occur along a reflection symmetric line, we reduce the zero mode equation to a one variable ordinary differential equation on that line. We can then compare to a harmonic oscillator in one dimension.

The one-dimensional harmonic oscillator is described by the Hamiltonian:
\begin{equation}
H = \frac{\hat{p}^2}{2m} + \frac{1}{2}(m\omega^2)x^2,
\end{equation}
whose $n_{th}$ eigenstate $\phi_n(x)$ satisfies the eigen-equation:
\beqn
-\frac{\hbar^2}{2m}\frac{d^2\phi_n}{dx^2}+\frac{1}{2}(m\omega^2)x^2\phi_n^2 = E_n\phi_n.
\eeqn
which can be transformed into the standard \emph{Sturm-Liouville form}, with dimensionless parameters $\alpha\equiv\sqrt{\hbar/(m\omega)}$, $\epsilon\equiv E/(\hbar\omega/2)$ and $u\equiv x/\alpha$:
\beqn
\frac{d}{du}[p(u)\frac{d\phi_n(u)}{du}] + \left(q(u) + \epsilon\omega(u)\right)\phi_n(u) = 0,\label{sturm-liouville}
\eeqn
where
\beqn
p(u) = 1,\quad q(u)=-u^2,\quad \omega(u)=1.
\eeqn

The normalizable solutions of Eqn.~(\ref{sturm-liouville}) are given by,
\beqn
\phi_n(u) = N_n \mathcal{H}_n(u)e^{-\frac{1}{2}u^2},\label{harmonic-os-eq}
\eeqn
where $N_n$=$(\frac{m\omega}{\pi\hbar})^{\frac{1}{4}}(2^nn!)^{-\frac{1}{2}}$ is the normalization factor and $\mathcal{H}_n$ is the $n_{th}$ Hermite polynomial. Therefore we see that for the harmonic oscillator, the number of zeros of the $n$-th excited eigenstate is $n$, and the parity of the $n_{th}$ eigenstate $\psi_n$ alternates as $(-1)^n$. Such oscillatory behavior is a generic feature for Sturm-Liouville type differential equations Eqn.~(\ref{sturm-liouville}) on the interval where $p(u)$ and $\omega(u)$ are positive \cite{ODE_book}.

We have observed a similar alternating parity and increasing number of zeros of eigenstates at higher magic angles in the chiral model, as discussed in Section~\ref{sec:exact_inversion_sym} and Section~\ref{nodesofzeromode}. The problem of the chiral twisted bilayer graphene model is more difficult. One reason is that it is a two-variable differential equation. To make progress, we utilize the symmetry of the problem to reduce the problem to one variable.

We starting by reviewing the zero mode equation, and see how symmetry helps reduce the dimension of the problem. We first recall the zero mode equation from Eqn.~(\ref{zeromodewavefunction}) and Eqn.~(\ref{AshvinHamiltonian}):
\beqn
-i\bar\partial(i\mG(\bm r)\Phi_{\bm k}(\bm r)) = -\eta\alpha U_{\phi}(\bm r)\mG(-\bm r)\Phi_{\bm k}(\bm r).\label{recallzeromodeeqn}
\eeqn
By using the lowest Landau level condition that the quantum Hall wavefunction $\Phi_{\bm k}$ satisfies,
\beqn
\bar\partial\Phi_{\bm k} = -\frac{z}{2}\Phi_{\bm k},\label{LLLcondition}
\eeqn
we arrive at the zero mode equation that the function $\mG(\bm r)$ must satisfy:
\beqn
(\bar\partial - \frac{z}{2}) \mG(\bm r) + \eta \alpha U_\phi(\bm r) \mG(-\bm r)=0,\label{36}
\eeqn
which is subject to the boundary condition Eqn.~(\ref{boundary-condition-G}).

We note that due to the mirror symmetry $\mM_y$ of the problem, both $\mG(x,y)$ and $\mG^*(x,-y)$ are zero mode solutions of Eqn.~(\ref{36}). By utilizing the global $U(1)$ phase degree of freedom of wavefunction, one can always choose $\mG(\bm r=\bm 0)$ to be a purely real number, thereby constraining $\mG(x,0)$ to be a real function. We have already used this property for $\rho(\bm r)$ in Section~\ref{nodesofzeromode}, and plotted its real and imaginary part on the reflection symmetric line in FIG.~\ref{plotaline}.

Here we denote the real and imaginary parts of $\mG(x,0)$ as $\mathcal{R}(x)$ and $\mathcal{I}(x)$ respectively. Although the imaginary part vanishes identically at $y$=$0$, its $y-$direction derivative $(\p_y\mathcal{I})(x)$$\equiv$$\partial_y\mathcal{I}(x,y)|_{y=0}$ does not. We end up with the following:
\beqn
\mathcal{I}(x) = 0,\quad\mathcal{R}(x)\neq0,\quad(\partial_y\mathcal{I})(x)\neq0.
\eeqn

The zero mode equation Eqn.~(\ref{36}) is now rewritten as:
\begin{equation}
\partial_x \mathcal{R}(x) - (\partial_y \mathcal{I})(x) - \frac{x}{2}\mathcal{R}(x) + \eta \alpha U_{\phi}(x) \mathcal{R}(-x) = 0,\label{37}
\end{equation}
subject to the boundary condition Eqn.~(\ref{boundary-condition-G}) which, when reduced to the $y$=$0$ line, becomes:
\beqn
\mathcal{R}(x+\sqrt3a) &=& -\mathcal{R}(x),\label{bc}\\
(\p_y\mathcal{I})(x+\sqrt3a) &=& -(\p_y\mathcal{I})(x) + \frac{\sqrt3a}{2} \mathcal{R}(x),\nonumber
\eeqn
where $a$ is the length of the \mr primitive lattice vectors. In the unit $S$=$1$ we have been using, its value is $a^2$=$4\pi/\sqrt3$.

The derivation so far is exact. The differential equation Eqn.~(\ref{37}) and its boundary conditions Eqn.~(\ref{bc}) contain the full information of the nodes in the problem. The difficulty of solving Eqn.~(\ref{37}) is that it is a two-variable differential equation. To make progress, we now do approximation on $(\p_y\mathcal{I})$ to eliminate one variable.

It is interesting to observe that $-\frac{x}{2}\mathcal{R}(x)$ satisfies the same boundary condition as $(\partial_y \mathcal{I})(x)$. In the following, we will approximate:
\beqn
(\p_x\mathcal{I})(x) \approx -\frac{x}{2}\mathcal{R}(x).\label{assump}
\eeqn
Under this assumption, the differential equation simplifies dramatically, and becomes a one-variable ordinary differential equation:
\begin{eqnarray}
\frac{d}{dx} \mathcal{R}(x)+ \eta \alpha U_{\phi}(x) \mathcal{R}(-x) = 0,\label{38}
\end{eqnarray}
which can be rewritten into a second order form:
\begin{eqnarray}
-\frac{d}{dx}\left(\frac{1}{U_\phi(x)}\frac{d\mathcal{R}(x)}{dx}\right) &=& \alpha^2 U_{\phi}(-x)\mathcal{R}(x).\nonumber\\
\mathcal{R}(x+\sqrt3a) &=& -\mathcal{R}(x).\label{39}
\end{eqnarray}

Hence we have brought the zero mode equation on the reflection symmetric line into the Sturm-Liouville form Eqn.~(\ref{sturm-liouville}) under the approximation shown in Eqn.~(\ref{assump}).

Suppose we have two solutions $\mathcal{R}_{1,2}$ of Eqn.~(\ref{39}), which corresponds to two magic angles $\alpha_{1,2,}$ with $\alpha_1<\alpha_2$. From Eqn.~(\ref{39}), we deduce that,
\beqn
[U^{-1}_{\phi}(\mathcal{R}_1\mathcal{R'}_2-\mathcal{R'}_1\mathcal{R}_2)]' = (\alpha_1^2-\alpha_2^2)U_{\phi}(-x)\mathcal{R}_1\mathcal{R}_2,\nonumber
\eeqn
where we have implicitly suppressed the argument $x$ in $U^{-1}_{\phi}$, $\mathcal{R}_{12}$ and their derivatives. Now, consider a region spanned $[x_a,x_b]$. The integration of the above equation in this region yields:
\begin{widetext}
\beqn
[U^{-1}_{\phi}(x)(\mathcal{R}_1(x)\mathcal{R'}_2(x)-\mathcal{R'}_1(x)\mathcal{R}_2(x))]|^{x_b}_{x_a} = (\alpha_1^2-\alpha_2^2)\int_{x_a}^{x_b}d\zeta U_{\phi}(-\zeta)\mathcal{R}_1(\zeta)\mathcal{R}_2(\zeta).\label{integralform}
\eeqn
\end{widetext}

It then follows from the theory of differential equations \cite{ODE_book}, in the parameter region $x\in[x_a,x_b]$ that $U_{\phi}(\pm x)>0$ or $U_{\phi}(\pm x)<0$, the nodes of two consecutive solutions must oscillate; otherwise it leads to contradiction with Eqn.~(\ref{integralform}). We emphasize that our argument is based on the assumption Eqn.~(\ref{assump}), and we can only argue for the node oscillation in the regions where $U_{\phi}(\pm x)$ are both positive or negative. This argument shows that in general, there should be more zeros at higher magic angles.

\bibliography{TBG0.bib}
\end{document}